\documentclass[12pt]{amsart}
\usepackage[margin=1in]{geometry}
\usepackage{lmodern}
\usepackage{mathptmx}
\usepackage{mathtools}
\usepackage{amsthm}
\usepackage{amssymb}
\usepackage{amsrefs}
\usepackage{graphicx}
\usepackage{tikz}
	\usetikzlibrary{decorations.pathreplacing}
	\usetikzlibrary{patterns}
\usepackage{caption}
\usepackage[usestackEOL]{stackengine}
\usepackage{scalerel}
\usepackage{calc}
\usepackage{accents}
\usepackage[new]{old-arrows}
\usepackage{subcaption}

\newcommand{\decom}{{\mathcal D}}
\newcommand{\cd}{\cdot}
\newcommand{\ra}{\rightarrow}
\newcommand{\pr}{\prime}
\newcommand{\de}{\partial}

\newcommand{\R}{\mathbb{R}}

\newcommand{\lbar}[1]{\overline{#1}}
\DeclareMathOperator{\Id}{Id}

\DeclareMathOperator{\ocinv}{\mathcal{R}}

\newcommand{\dep}{D}
\newcommand{\be}{\begin{equation}}
\newcommand{\ee}{\end{equation}}

\newcommand{\Swap}{{\rm Swap}}

\newcommand{\hfn}{h}

\DeclareMathAlphabet{\mathcal}{OMS}{cmsy}{m}{n}

\def\myupbracefill#1{\rotatebox{90}{\stretchto{\{}{#1}}}
\def\rlwd{.5pt}
\newcommand\notate[4][B]{
  \if B#1\else\def\myupbracefill##1{}\fi
  \def\useanchorwidth{T}
  \setbox0=\hbox{$\displaystyle#2$}
  \def\stackalignment{c}\stackunder[-6pt]{
    \def\stackalignment{c}\stackunder[-1.5pt]{
      \stackunder[2pt]{\strut $\displaystyle#2$}{\myupbracefill{\wd0}}}{
    \rule{\rlwd}{#3\baselineskip}}}{
  \strut\kern13pt$\rightarrow$\smash{\rlap{$~\displaystyle#4$}}}
}

\newtheorem{theorem}{Theorem}[section]
\newtheorem{corollary}[theorem]{Corollary}
\newtheorem{lemma}[theorem]{Lemma}

\newtheorem{defin}[theorem]{Definition}
\newtheorem{assumption}[theorem]{Assumption}
\begin{document}

\title{The Group Structure of Quantum Cellular Automata}

\author{Michael Freedman}
\address{\hskip-\parindent
	Michael Freedman\\
    Microsoft Research, Station Q, and Department of Mathematics\\
    University of California, Santa Barbara\\
    Santa Barbara, CA 93106\\}
\email{mfreedman@math.ucsb.edu}

\author{Jeongwan Haah}
\address{\hskip-\parindent
Jeongwan Haah\\Microsoft Quantum and Microsoft Research\\ Redmond, WA 98052, USA}
\email{jwhaah@microsoft.com}

\author{Matthew B.~Hastings}
\address{\hskip-\parindent
	Matthew Hastings\\
    Microsoft Research, Station Q\\
    University of California, Santa Barbara\\
    Santa Barbara, CA 93106\\}
\email{mahastin@microsoft.com}

\begin{abstract}
We consider the group structure of quantum cellular automata (QCA) modulo circuits 
and show that it is abelian even without assuming the presence of ancillas, 
at least for most reasonable choices of control space; 
this is a corollary of a general method of ancilla removal.
Further, we show how to define a group of QCA that is well-defined without needing to use families, 
by showing how to construct a coherent family containing an arbitrary finite QCA;
the coherent family consists of QCA on progressively finer systems of qudits
where any two members are related by a shallow quantum circuit.
This construction applied to translation invariant QCA shows that
all translation invariant QCA in three dimensions 
and all translation invariant Clifford QCA in any dimension
are coherent.
\end{abstract}
\maketitle

The group of quantum cellular automata (QCA) modulo quantum circuits is an abelian group\cites{fh,hfh} (see Appendix).  While this group 
has been fully understood in one dimension using an index theory\cites{Gross_2012,fermionGNVW1,fermionGNVW2}, and 
there has been much recent progress in higher dimensions\cites{hfh,haah2019clifford}, there are still some unsatisfactory 
foundational issues.  For one issue, the group was only known to be abelian with additional ancillas present and the case without ancillas was not understood.  For another issue,  multiplication of two QCAs will necessarily increase the range of the QCA until eventually (on any finite system) some finite product of QCAs may give a QCA whose range is comparable to the 
system size, thus rendering the very idea of locality in that QCA moot.  One way of dealing with this issue is to consider infinite families of QCAs with a fixed control space and decreasing range of the QCA so that for any finite product of families all but finitely many QCA in the family will have range small compared to the system size.  However, this is not completely satisfactory since an arbitrary family may involve a sequence of QCAs that are ``unrelated" to each other; for example, various indices might fluctuate arbitrarily from one member of the family to the next.

In this paper, we present results to resolve these issues.  
First, we show how to remove ancillas (for most reasonable choices of control space) so that the group of QCA on those control spaces is abelian even without ancilla.
This is a corollary of a theorem that 
if the action of a quantum circuit is the identity on ancilla
then the action can be realized by a slightly deeper quantum circuit
whose gates are not supported on the ancilla.
Second, we show how, given a single fixed QCA, 
to create a ``coherent" family of QCA\cite{fh} containing that ``mother'' QCA.
This allows us to define a useful group structure given even a single QCA,
where the natural composition of mother QCA is compatible 
with the elementwise composition of the constructed coherent family.
The coherent family is parallel to entanglement renormalization groups of topological many-body states,
and our construction can be colloquially regarded as a way to construct a ``UV theory.''
Lastly, we examine translation invariant QCA 
that comes in an obvious family of QCA under periodic boundary conditions.
Two classes of QCA are shown to be entanglement renormalization group fixed points:
Every translation invariant Clifford QCA (one that maps a Pauli operator to a tensor product of Pauli operators)
in any dimension
and every translation invariant QCA in three dimensions
defines a coherent family under periodic boundary conditions.

\section{Definitions}\label{sec:def}
\subsection{Naive Definitions}
Let us first define a QCA, a quantum circuit, and the notion of a control space.  
We will give two definitions in each case, first a simple definition when there is a finite number of ``sites" and then a more general definition.
Here we are defining what we may term a ``naive" object, in that we will give a single QCA or circuit, rather than a family.

The overall set up is: there is a set $\{\mathcal{H}_i\}$, $i \in I$, of finite dimensional Hilbert spaces indexed by a discrete set $I$. 
The index $i$ labelling the finite dimensional Hilbert spaces is referred to as a ``site" index.
In turn, each finite dimensional Hilbert space $\mathcal{H}_i$ is a tensor product of some number of additional Hilbert spaces, termed ``qudits" associate with that site.  This terminology will be used later when discussing ancillas.

We introduce some metric ${\rm dist}(\cdot,\cdot)$ between sites.  This metric may be derived from a control space $X$ (where a control space is a fixed smooth manifold with a metric, or alternatively a simplicial complex), given a map $\mathcal{I}: I \ra X$ from sites to points in the control space.  
The control space may have boundary and may be noncompact.
We assume the map is locally finite, i.e. for $C$ compact, $C \in X$, $\mathcal{I}^{-1}(C)$ is finite.

If the set of sites is finite, then a QCA is a $*$-automorphism of the algebra of operators on this Hilbert space, subject to a locality condition explained below, and any such QCA can be expressed as conjugation by some unitary, also with a locality condition.

More generally, if the set is not finite,
we consider the net of operators on $\underset{i \in J \subset I}{\otimes} \mathcal{H}_i$, where $J$ is finite. Implicitly operators can always be extended to larger finite $J^\pr \subset I$ by tensoring with id over $J^\pr \setminus J$. The support $\mathrm{supp}(\mathcal{O})$ of an operator $\mathcal{O}$ is the smallest $J$ over which it may be written.
a QCA is a $\ast$-automorphism $\alpha$ of the associated net of Endomorphism algebras $\{\underset{i \in J}{\otimes} \mathrm{End}(\mathcal{H}_i)\}$.

In both finite and infinite cases, there is a 
a geometric condition that there is a range $R \in \R^+$ so that $\alpha$ is $R$-local, i.e. for all $\mathcal{O}$, $\mathrm{supp}(\alpha \mathcal{O}) \subset \mathcal{N}_R(\mathrm{supp}(\mathcal{O})$, the radius $R$ neighborhood.
We will say either that the QCA has range $R$, or that it is an $R$-QCA.

Next we define a circuit more precisely.
Informally, a quantum circuit is some sequence of ``gates", each being a unitary acting on a set of bounded range.
Formally, a quantum circuit is a pair consisting of a unitary and a ``circuit decomposition" of that unitary in terms of gates.
Such a circuit decomposition of a unitary $U$ consists of writing $U$ in the form
\be
\label{qcecomp}
U=U_\dep\circ U_{\dep-1} \circ \ldots \circ U_1,
\ee
where
$\dep\geq 1$ is some integer called the ``depth" of the quantum circuit, and where further
each unitary $U_a$ is written in the form
$$U_a=\prod_{S\in G_a} U_{S,a},$$
where $G_a$ is a collection of disjoint sets of sites and where $U_{S,a}$ is a unitary supported on $S$.  
Disjointedness means that the ordering within the product for $U_a$ is immaterial.
Each $U_{S,a}$ is called a gate on set $S$.
We require that the diameter of all $S\in G_a$ for all $a$ be bounded by a constant called the ``range" of the gates.
The integer $a$ is referred to as labelling the ``round" of the circuit.  The range of the circuit is defined to be the range of the gates multiplied by the number of rounds.  We say that the circuit ``implements" the unitary $U$.
Hence such a circuit acts as a QCA by conjugation: $\mathcal{O} \ra U \mathcal{O} U^\dagger$.  
Informally, then, we may say that some QCA ``is equal to a circuit" if the circuit acts as that QCA by conjugation.
We emphasize that the range of the resulting QCA is bounded by the range of the circuit, and it may in some cases be significantly smaller; see section \ref{circ}.

Finally, let us define the notion of an ancilla.
As mentioned above, the Hilbert space on each site is a tensor product of some number of Hilbert spaces termed qudits.
Each qudit will be referred to as either ``physical" or ``ancilla".
We say that a quantum circuit ``implements a unitary $V$ using ancillas" if
the quantum circuit implements some unitary $U=V\otimes I$ where $V$ acts on the physical qudits and $I$ acts on the ancillas.

\subsection{Families and Coherent Families}
QCA have a natural group structure, in that two QCA can be multiplied by composing them.  However, this composition may increase the range, so that the product of an $R_1$-QCA with an $R_2$-QCA may be an $(R_1+R_2)$-QCA.  This presents no problem if the control space is infinite.  For example, Ref.~\cites{Gross_2012} studied an infinite real line as the control space, developing an index theory.  However, if the control space has finite diameter, a finite number of compositions of QCA may give a QCA whose range is comparable to the injectivity radius of the control space itself, and all locality is lost.

One way to resolve this is to consider families of QCA on a fixed control space.  A family $\lbar{\alpha}$ is a sequence of QCA, $\alpha_i$, with $i=0,1,2,\ldots$, each with some range $R_i$, with $R_i\rightarrow 0$ are $i\rightarrow \infty$.
Note that we will use a bar to denote a family.
While all QCA in the family have the same control space, different QCA in the family may have different sets of sites.
The product of two families $\lbar{\alpha}\circ \lbar{\beta}$ is the family defined by the sequence of QCA $\alpha_i \circ \beta_i$.
Then, for any finite product of families, for all but finitely many QCA in the resulting family the range $R_i$ will be small compared to the injectivity radius of the control space.
This allows the index theory of Refs.~\cites{Gross_2012,fh} to serve as an obstruction to writing a family of QCA as a family of quantum circuits.
We define a family of quantum circuits to be a sequence of quantum circuits, with the depth of all circuits in the family bounded, and with the range of the gates in the quantum circuit tending to zero.
Such a family of quantum circuits will be called a finite depth quantum circuit (fdqc).

Remark: an alternative scaling of distance is to take all QCA in the family to have the same range $R$ but to rescale the metric on $X$ so that the injectivity radius of $X$ diverges.  This differs from the definition here simply by an overall scaling of the metric.

Remark on notion of bounded range quantum circuit: one might be tempted instead to consider families of
bounded range quantum circuits (brqc), defined as a sequence of quantum circuits over $X$ so that the range of the $i$-th quantum circuit tends to $0$ as $i\rightarrow \infty$ (the depth is allowed to diverge so long as the product of range and depth is bounded). 
More generally, even, one might consider circuits such that the range of the lightcone tends to zero (one may define the range of the lightcone in an intuitive way as the maximum range of the QCA implemented by the circuit assuming a generic choice of gates; we give a more precise definition later).
However, assuming the existence of an appropriate handle decomposition on the control space, all these notions are equivalent.  To give an intuitive idea of the equivalence, consider the case in one-dimension as shown in Fig.~\ref{1dh}.
We see here a circuit of depth $3$ with each gate having range $1$.  However, it is shown how to regroup it into a circuit of depth $2$, with each gate having range $O(1)$.  In general, assuming an appropriate handle decomposition, we can regroup a circuit to depth $d+1$ if the control space is $d$-dimensional.  We explain this in more detail in section \ref{circ}.

\begin{figure}
\includegraphics[width=6in]{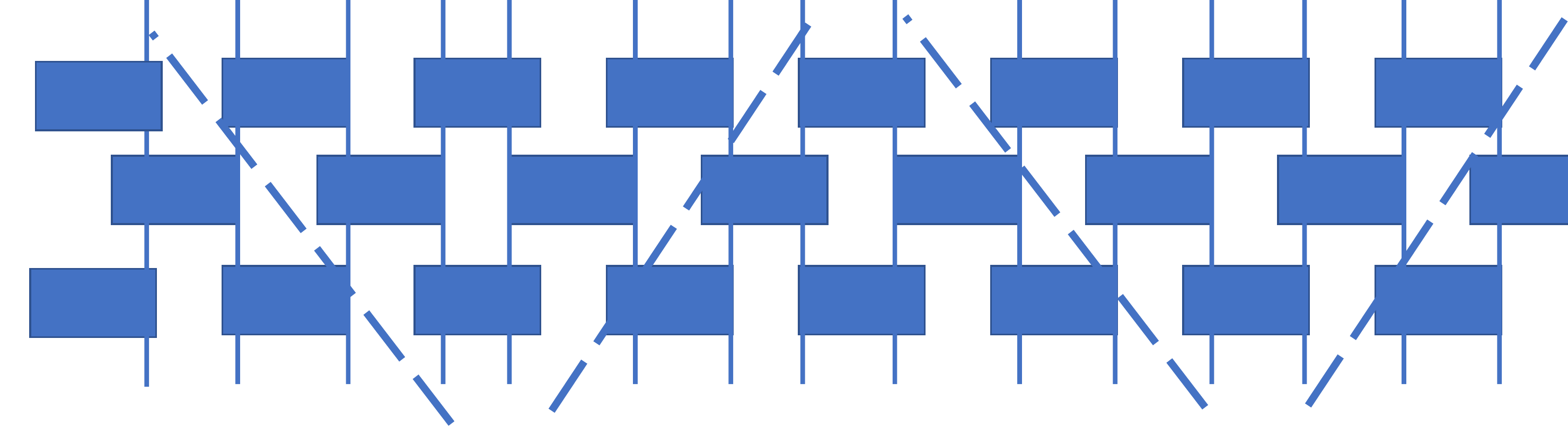}
\caption{A quantum circuit.  Horizontal axis is position, and vertical axis labels the round, increasing upwards.  Each rectangle is a gate, and lines represent qudits.  The pattern repeats indefinitely.  The dashed lines show how to combine a set of gates into a single unitary: all gates between a pair of neighboring dashed lines can be combined into a single unitary.  This rewrites the circuit as depth two but with gates having diameter $5$ in this example.}
\label{1dh}
\end{figure}

However, there is a problem with the notion of a family.  One may imagine a family composed of QCA which are ``unrelated to each other" in some sense.  For example, take the control space to be a circle and let QCA $\alpha_i$ act on a Hilbert space built from $2^i$ sites, with one qudit per site and with sites uniformly spaced around the circle so distance
$2^{-i}$ between neighboring sites.  If we choose for example $\alpha_i$ for even $i$ to be the identity QCA and $\alpha_i$ for odd $i$ to be a shift right by one, then the GNVW index fluctuates between different QCA in the family, which is undesirable.
Indeed, the QCA in the sequence could be even more arbitrary than this example; for example, it is quite possible to have the GNVW index diverge with $i$.

As a resolution,
the following definition of a ``coherent" family was proposed in Ref.~\cite{fh}
 (an earlier version called this family ``uniform").
To define 
the notion of coherent family, we first need a notion of path equivalence and stable path equivalence
 as follows:
 \begin{defin}
 Two QCA $\alpha,\beta$ are $R'$ path equivalent if there exists a continuous path of QCA with range at most $R'$ from $\alpha$ to $\beta$.
 \end{defin}
 \begin{defin}
 Two QCA $\alpha,\beta$ are stably $R'$ equivalent if one can tensor with additional ancilla degrees of freedom such that
 $\alpha \otimes \Id$ is $R'$ path equivalent to $\beta \otimes \Id$, where $\Id$ denotes the identity QCA and in this case $\Id$ acts on the additional degrees of freedom.
 Here, $\alpha,\beta$ may act on different Hilbert spaces with different sites (although they have the same control space), so that one may tensor $\alpha$ with $\Id$ acting on one set of ancillas and tensor $\beta$ with $\Id$ acting on a different set of ancillas.  There must be a one-to-one correspondence between sites acted on by $\alpha\otimes \Id$ with those acted on by $\beta\otimes \Id$, with corresponding sites having the same finite dimensional Hilbert space and the same image under the map to the control space.  In particular, sites which are ancillas for $\alpha \otimes \Id$ might not be ancillas for $\beta\otimes \Id$.
 \end{defin} 
 
 To define a coherent family,
we take $R_i=R \cdot 2^{-i}$ for some $R$ and then define:
 \begin{defin}
 \label{cohfam}
Such a family $\lbar{\alpha}$ of QCA is ``path coherent" (or simply ``coherent") if there exists a constant $c$ such that for all $i$, QCA $\alpha_i$ is stably $cR_i$ path equivalent to $\alpha_{i+1}$.
\end{defin}

 We emphasize that this definition implies also that for any $j>0$ there is some constant $c_j$ such that
for all $i$ QCA $\alpha_i$ is stably $c_jR_i$-equivalent to $\alpha_{i+j}$.
Proof: indeed, we can take $c_j=c+c/2+c/4+\ldots c/2^{j-1}\leq 2c$, by composing the circuits that map $\alpha_k$ to $\alpha_{k+1}$ for $k=i,i+1,\ldots,i+j-1$.

Further, it is clear that our assumption that $R_i=2^{-i}$ could be replaced any exponential decay, i.e., with the assumption that $R_i=x^{-i}$ for any $x>1$.
If $x>2$, then such a family is also a family of QCA with range $R_i=2^{-i}$.  If $x<2$, then one can pass to a subfamily and obtain a family with range $R_i=2^{-i}$.  So long as $R_i$ decays exponentially, then passing
to a subfamily does not affect the coherence of the family, essentially for the reasons explained in the above paragraph, taking now $c_j=c+c/x+c/x^2+\ldots+c/x^{j-1}$.

In fact, one may define an even stronger notion of coherence where QCA are related by circuits rather than by paths, as follows. 
 \begin{defin}
 Two QCA $\alpha,\beta$ are $R'$ circuit equivalent if there exists a quantum circuit of range $R'$ such that $\alpha$ is equal to $\beta$ followed by the quantum circuit, i.e., acting by conjugation by the quantum circuit implements QCA $\alpha \circ \beta^{-1}$.
 \end{defin}
 \begin{defin}
 Two QCA $\alpha,\beta$ are stably $R'$ circuit equivalent if one can tensor with additional ancilla degrees of freedom such that
 $\alpha \otimes \Id$ is $R'$ circuit equivalent to $\beta \otimes \Id$, where $\Id$ acts on the additional degrees of freedom.
 \end{defin} 
 \begin{defin}
 \label{circcohfam}
A family $\lbar{\alpha}$ of QCA is``circuit coherent" if there exists a constant $c$ such that for all $i$, QCA $\alpha_i$ is stably $cR_i$ circuit equivalent to $\alpha_{i+1}$.
\end{defin}

Note that every circuit coherent family is path coherent.  Proof: choose a path joining each gate in the circuit to the identity.

With this definition, the example above of shifts is not coherent (indeed, it is not even path coherent) but if we choose each QCA to shift right by one, then the family  is circuit coherent.

Given this definition of coherent families, the reader might be tempted to assume that one should only consider coherent families and that ``incoherent families" (i.e., those which are not coherent) are somehow artificial.  However, consider the following natural way of constructing a family.  Consider a translation invariant QCA acting on an infinite lattice of sites in a hypercubic lattice in $n$ spatial dimensions.  Then, there is an obvious way to define a family of QCA acting on a finite lattice, with the control space being the $n$-torus.  Assume the translation invariance holds for any translation by a single site, in any of the $n$ different spatial directions.  

For all sufficiently large $i$, let QCA $\alpha_i$ act on a lattice of $(m^{i})^n$ sites for some integer $m>1$.  The sites will be arranged in a hypercubic lattice within the torus, with distance $m^{-i}$ between neighboring sites.  We define $\alpha_i$ in the obvious way from the translation invariant QCA.  For example, the family of shifts where each QCA shifts right by one can be obtained from a translation invariant QCA which shifts right by one on an infinite system.
Remark: here we require that $i$ be sufficiently large, as to define the QCA we need $m^i$ large enough compared to the range of the translation invariant QCA.

Then, it is not clear whether or not this obvious choice of a family is a coherent family.  
Indeed, one can give explicit examples where the family is not expected to be coherent.  For example, in Ref.~\cite{hfh} a translation invariant QCA $\alpha_{WW}$ was constructed in three dimensions whose square is a quantum circuit but strong evidence was given that $\alpha_{WW}$ itself could not be implemented by a quantum circuit.  
A generalization of this was given in Ref.~\cite{haah2019clifford} where QCA $\gamma$ were constructed whose {\it fourth} power was a quantum circuit but it is believed that $\gamma,\gamma^2,\gamma^3$ cannot be implement by quantum circuits.
Consider a four dimensional translation invariant QCA $\beta$ given by stacking copies of $\gamma$, so that sites in the fourth dimension are labelled by some integer, and for each choice of this integer we have an independent copy of $\gamma$.  Finally, choose $m=3$.  Then, $m^i=1,3,9,\ldots$.  Under the assumption that $\gamma,\gamma^2,\gamma^3$ are not quantum circuits, this is not a coherent family.  Proof: dimensionally reduce, i.e., ignore the requirement of locality in the fourth dimension.  Then, for varying $m$ we have $1,3,9,\ldots$ copies of $\gamma$.  Mod $4$ this is $1,-1,1,-1,\ldots$ copies of $\gamma$.  However, by assumption, $\gamma$ is not related to $\gamma^{-1}$ by a quantum circuit. 

We further discuss families of translation invariant QCA in section \ref{cft} and give some sufficient conditions for the family to be coherent.

This phenomenon may be reminiscent of something that occurs in Hamiltonians which are sums of commuting projectors with topologically ordered ground states.  If one takes a toric code, then for any fixed manifold, for any 
cellulation ${\mathcal C}$, one may construct an isometry from the ground state subspace of the Hamiltonian on that cellulation ${\mathcal C}$ to the ground
state subspace of the Hamiltonian on a refinement ${\mathcal C}'$ of that cellulation by tensoring with additional
additional degrees of freedom in some product state and applying a local quantum circuit\cite{aguado2008entanglement}.  Thus, the ground states of this Hamiltonian have some similar ``coherence" property.  However, the ground state structure of the cubic code\cite{haah2014bifurcation} is much more complicated and does not have this property.  This is related to a difference in the renormalization structure of the Hamiltonian, where the toric code Hamiltonian renormalizes to itself plus topologically trivial terms but the cubic code Hamiltonian renormalizes to itself plus terms describing an additional topological phase.

\section{Ancilla Removal in Quantum Circuits}
\label{arqc}
\subsection{Overview}
As remarked before,
several authors have noted that the group of quantum cellular automata modulo quantum circuits is an abelian group in the presence of ancilla.  That is, given any two QCA $\alpha,\beta$, both with some bounded range $R$, the product $\alpha \circ \beta \circ \alpha^{-1} \circ \beta^{-1}$ can be written as a circuit with ancillas, i.e., by adding some additional ancilla degrees of freedom, we can find a quantum circuit of bounded depth and range that implements the unitary $\Bigl( \alpha \circ \beta \circ \alpha^{-1} \circ \beta^{-1}\Bigr) \otimes I$, where the tensor product is between the physical and ancilla degrees of freedom, so the circuit acts as $\alpha \circ \beta \circ \alpha^{-1} \circ \beta^{-1}$ on the physical degrees of freedom and as the identity on the ancilla degrees of freedom.

This raises the question: if we do not allow ancillas, is the group still abelian?  We answer this question in the affirmative (at least for certain choices of control space) by showing a more general result: for certain choices of control space (defined below), given any circuit that acts trivially on the ancilla degrees of freedom, there is a circuit implements the same unitary on the physical degrees of freedom without using any ancilla degrees of freedom.  The latter circuit has an increased depth and range, but they are bounded by some function of the depth and range of the original circuit.

In this note qudits will be assumed to have the same dimension.  Given the construction for ancilla removal below, the reader will see how to generalize it to cases with varying dimension though (assumption \ref{a1} below must be generalized so that it holds separately for each prime dividing the dimension of the qudits).

The basic idea of the construction involves taking a given decomposition of the circuit and constructing a new decomposition.
A trivial example of this is that
given any circuit, we can increase the depth by $1$ by taking one of the $U_a$ to be the identity operator, i.e., given a circuit with given $U_a$ and given depth $\dep$, we could define a new circuit $U'$ with depth $\dep+1$ and corresponding unitaries $U'_a$ with $U'_a=U_a$ for $a\leq \dep$ and $U_{\dep+1}=I$.
However, more interesting and useful choices of decomposition are possible.
We will explain how, given a circuit, to change the decomposition so that in each round, the density of the gates decreases (roughly meaning that only a small fraction of the qudits are in the support of any gate; we make this more precise later), at the cost of increasing the depth of the circuit.  We use this modified decomposition to show how it is possible to remove ancillas from a quantum circuit, turning a circuit that uses ancillas into one without ancillas, at the cost of increased depth and under certain assumptions on the ancilla density.

\subsection{Simplest Construction}
The basic idea of the decomposition that we will use can be seen in the following example that we present pictorially.  
After giving this example, we give a more general construction with formal definitions.
We consider a one-dimensional system.  
Sites are labelled by integers, and the distance between sites is simply the absolute value of the difference between the corresponding integers (modulo some large integer if one wants to take a finite periodic system).  All sites have two qudits, one physical and one ancilla.  (It should be apparent that here we have scaled the distance between sites to $1$ for notational simplicity; this is a different scaling than taking the control space to have a fixed size with distance between sites small.)

We consider a circuit with some bounded depth.
Imagine writing gates of the circuit in a two dimensional array, with horizontal axis labelling the site coordinate, and vertical axis labelling the round of the gate.  Now consider the decomposition of
Fig.~\ref{isfig}.  
The numbers $1,2,3,4$ represent ``blocks" of qudits, each containing a number of qudits which is some constant multiple times the range of the circuit.
The individual gates are not shown in this figure.  Rather, each block labelled $A,B,C,D$ contains many gates horizontally and vertically.  Each gate is entirely contained within a single block of gates; the angled line represents a ``lightcone".

\begin{figure}
\includegraphics[width=6in]{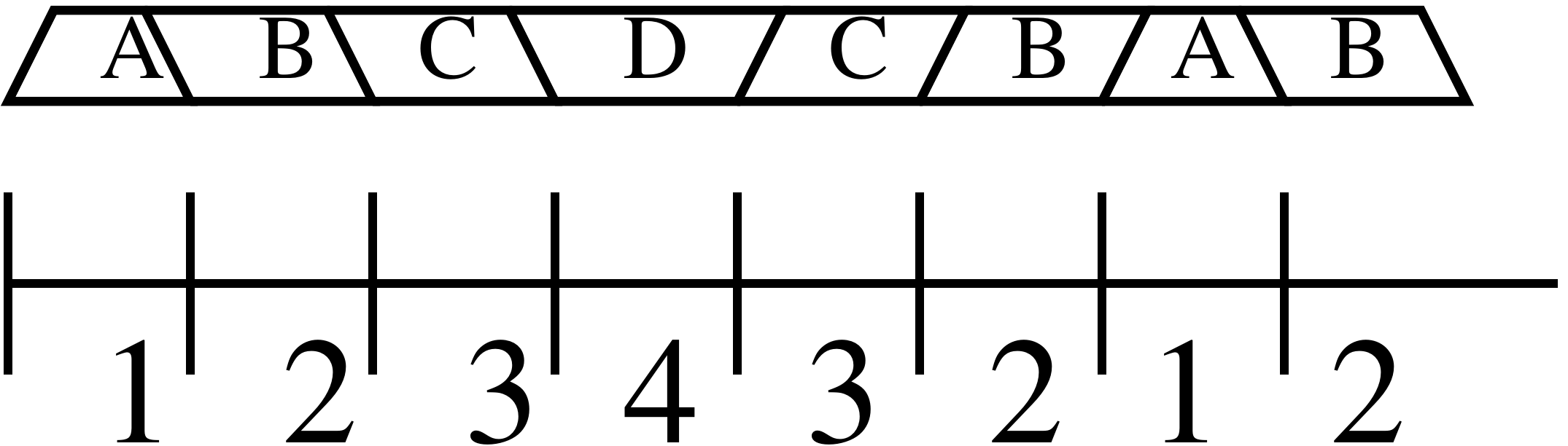}
\caption{Horizontal axis is position.  Vertical axis labels the round, increasing upwards.  The bottom line labelled by numbers represents the qudits, with each number $1,2,3,4$ representing a set of qudits, while the top of the figure represents gates.
The figure repeats, so that to the right of the last box labelled 2 is another box labelled 3, then 4, then 3,2,1, and so on, and similarly to the right of the last box labelled B are boxes labelled C,D,C, and so on.  }
\label{isfig}
\end{figure}

Let $V_a$ for $a\in\{A,B,C,D\}$ be the product of the unitaries in all blocks labelled $a$.  Then, the circuit is equal to the product $V_D V_C V_B V_A$.  This product gives a different decomposition of the circuit, in which gates in blocks labelled $D$ are done in a later round than they would be in the original decomposition.  The depth of the circuit is increased by at most a factor of $4$.

This new decomposition gives rise to what one may pictorially describe as a staircase or a zig-zag: if, as a function of the coordinate of sites, one considers those rounds over which some gate is acting on that site, these will increase and decrease periodically in a zig-zag, with sites in a block labelled $4$ having gates at the latest rounds (i.e., they are in gates in blocks $D$) and those in a block labelled $1$ having gates in the earliest rounds (i.e., in gates in blocks $A,B$).

We then use an idea of ``borrowing" to remove ancillas: for each gate acting on an ancilla qudit in block $1$, we replace its action on that ancilla with the same action on a physical qudit in block $3$.  We similarly replace the action on ancillas in block $2$ with physical qudits in block $4$, ancillas in block $3$ with physical qudits in block $1$, and ancillas in block $4$ with physical qudits in block $2$.
Precisely, we construct a one-to-one pairing between ancilla qudits in block $1$ with physical qudits in block $3$, and similarly between ancillas in block $2$ with physical qudits in block $4$, and so on.  Such a pairing exists because we will take all blocks the same size.  
We pair each ancilla qudit in block $1$ with a physical qudit in a neighboring block $3$ (for example, the block $3$ immediately to the right).
Then, we conjugate $V_A$ by
a swap gate which swaps the ancillas in blocks $1$ with the corresponding physical qudits in blocks $3$.  
Similarly, we conjugate $V_B$ by a swap which swaps the ancillas in block $1$ with physical qudits in block $3$ and swaps the ancillas in block $2$ with physical qudits in blocks $4$; note crucially that $V_B$ does not act on the physical qudits in blocks $3$ or $4$.
We conjugate $V_C$ by a swap which swaps the ancillas in block $2$ with physical qudits in block $4$ and swaps the ancillas in block $3$ with physical qudits in block $1$.  We conjugate $V_D$ by a swap which swaps the ancillas in block $3$ with physical qudits in block $1$ and swaps the ancillas in block $4$ with physical qudits in block $2$.

Then one may show under the assumption that the original circuit acted trivially on the ancillas, that this replacement gives a new circuit that implements the same unitary.
We call this ``borrowing" because one ``borrows" the physical qudit in block $3$ to act as an ancilla for block $1$, and then returns it unchanged so later it can be used as a physical qudit.
The qudit is returned unchanged because the original circuit acted as the identity on the ancilla qudits.
Since the blocks have width that is proportional to the range of the quantum circuit, this requires increasing the range of the gates in the quantum circuit only by a factor proportional to the range of the whole circuit.

A similar kind of decomposition can be used when the control space is more than one-dimensional, so long as it has the form of a one-dimensional line ${\mathbb R}$ times some arbitrary space $X^\perp$ and so long
there is a certain symmetry under translation along the one-dimensional line.  Label coordinates in the control space
by a pair $(x,y)$ with $x\in {\mathbb R}$ and $y\in X^\perp$.  Then, we require that for all $x,y$, there be a site at $(x,y)$ if and only if there is one at $(x+1,y)$, and that both such sites (if they exist) have the same Hilbert space dimension.  If these conditions hold,
one simply uses the coordinate along the one-dimensional line to define the block decomposition.  One picks the swap gates to swap qudits at the same value of coordinate in the space $X^\perp$ but shifted coordinate along the one-dimensional line.

If we consider the more general setting in which only some of the sites have physical qudits, we can still make such a borrowing construction, given an assumption on the density of the physical qudits.  For example, suppose that sites labelled by integers equal to $0$ mod $K$ for some integer $K>0$ have one physical qudit, and all other sites have one ancilla qudit so that only a fraction $1/K$ of the qudits are physical. 
One way to proceed here is to iterate the above construction, removing half of the ancilla qudits at each step, i.e., suppose for simplicity that $K$ is a power of $2$ (if it is not, one may tensor with the identity on additional ancillas so that it is).  Consider a new system with one physical qudit and one ancilla qudit on each site, where the physical qudit is the tensor product of the original physical qudit and $K/2-1$ of the ancillas, and the ancilla is the tensor product of the remaining $K/2$ ancillas.  Applying the above construction, one can remove  half of the original qudits, and then repeat until no
ancillas are left.

An alternative way to handle the case in which some sites have no physical qudits is to construct a more complicated way of borrowing physical qudits to replace ancillas.  Such a more general construction (both in one-dimension and on more general control spaces) will be called a ``borrowing function" below.

\subsection{General Construction}
We now give a more general construction for other control spaces.
Given a quantum circuit $U$ and some decomposition of that circuit, and some pair $(S,a)$ where $S$ is a set and $a$ is an integer, define the ``forward lightcone" of $(S,a)$ inductively as follows.  It is the smallest set of pairs which contains $(S,a)$ and such that
if $U_{T,b}$ and $U_{T',b'}$ are in the decomposition of the circuit and $b'> b$ and $T'\cap T \neq \emptyset$ then $(T',b')$ is in the forward lightcone whenever $(T,b)$ is.

Given some decomposition of a quantum circuit, and given some function $\tau(\cdot)$ from pairs $(S,a)$ to integers, we say that this function $\tau(\cdot)$ (which we call a ``time function") is ``causal"
if given any pair of gates $U_{S,a}$ and $U_{T,b}$ with $b>a$ and $S\cap T \neq \emptyset$, then
$\tau(T,b)>\tau(S,a)$.
Note that for any causal time function, it follows that
given any pair of gates $U_{S,a}$ and $U_{T,b}$ with $(T,b)$ in the forward lightcone of $(S,a)$ we have that $\tau(T,b)>\tau(S,a)$.  Proof: if $(T,b)$ is in the forward lightcone of $(S,a)$ then there is some sequence $(S_1,a_1),(S_2,a_2),\ldots,(S_k,a_k)$ with $b>a_k>\ldots>a_1>a$ and $T\cap a_k\neq \emptyset$ and $a_{k}\cap a_{k-1} \neq \emptyset$ and $a_1\cap S\neq \emptyset$, and then the claim follows inductively.

We will say that a gate $U_{S,a}$ ``acts at time $t$" if $\tau(S,a)=t$.

We now show that
\begin{lemma}
\label{Vbdef}
Given some decomposition $\decom$ of a unitary $U$ and given some causal function $\tau$ with the range of $\tau$ being $1,\ldots,\tau_{max}$ for some arbitrary integer $\tau_{max}\geq 1$, then
$U=V_{\tau_{max}} \circ \ldots \circ V_1,$
where
\be
V_b=\prod_{(S,a) \; {\rm s. t.} \; \tau((S,a))=b} U_{S,a}.
\ee
That is, $V_b$ is the product of $U_{S,a}$ over all gates $S,a$ in the decomposition $\decom$ such that $\tau(S,a)=b$.
\begin{proof}
The proof is inductive in $\tau_{max}$.  
The base case $\tau_{max}=0$ is trivial since then $U$ is the identity and has no gates in its circuit decomposition.

Let $U$ have depth $\dep$.
Define $\tilde U$ by
$\tilde U=\tilde U_\dep\circ \ldots \circ \tilde U_1$, where
$\tilde U_a=\prod_{S\in G_a, \, \tau((S,a))<\tau_{max}} U_{S,a}.$
In words, $\tilde U$ has the same circuit decomposition as $U$, except that we remove all gates $U_{S,a}$ with $\tau((S,a))=\tau_{max}$, and $\tilde U_a$ is the same as $U_a$ except that we remove all gates $U_{S,a}$ with $\tau((S,a))=\tau_{max}$.

By assumption,
 there is no gate $U_{T,b}$ that is in the forward lightcone of any of the gates in $V_{\tau_{max}}$.
Hence, in the decomposition $U=U_\dep \circ \ldots \circ U_1$, the gates $U_{S,a}$ with $\tau((S,a))=\tau_{max}$ can be commuted to the left of all other gates, which shows that $U=V_{\tau{max}} \tilde U$.
By the inductive assumption, 
$\tilde U=V_{\tau_{max}-1} \circ \ldots \circ V_1$ since $\tau(\cdot)$ is also a causal time function for $\tilde U$.
Hence, $U=V_{\tau_{max}} \circ \ldots \circ V_1$.
\end{proof}
\end{lemma}

We claim that
\begin{lemma}
\label{tfun}
Let $\hfn: X\ra {\mathbb R}$ be a Lipschitz function with Lipschitz constant $K$.
Given such a function, define $\hfn(S)$ for any set $S$ to equal ${\rm min}_{x\in S} \hfn(x)$.

Then,
given a quantum circuit where the gates have range at most $r$, and given any set of sites $B$,
the time function
$\tau(S,a)=\lfloor \hfn(S) \rfloor+c a$
is causal for any $c\geq rK+1$.
\begin{proof}
Consider any pair of gates $U_{S,a}$ and $U_{T,b}$ with $b>a$ and $S\cap T \neq \emptyset$.
Then, by the Lipschitz condition, $\hfn(S)\leq \hfn(T)+rK$.
Hence, $\lfloor \hfn(S) \rfloor < \lfloor \hfn(T) \rfloor + rK+1$.
Hence for $c\geq rK+1$, $\lfloor \hfn(S)\rfloor+c a<\lfloor \hfn(T)\rfloor +c b$.
\end{proof}
\end{lemma}
We will call such a function $\hfn(\cdot)$ a ``height" function.
Remark: an example of such a height function that we will use later is
$\hfn(x)= {\rm dist}(x,B)$, where $B$ is any set of points in the control space.  This function has Lipschitz constant $1$.

We make one heuristic remark on this construction.  This time function $\lfloor  \hfn(S) \rfloor+c a$ can be understood as defining a new time depending both on the original ``time", i.e. $a$, as well as the position in ``space", i.e., the function $\hfn(S)$.  This is somewhat reminiscent of coordinate transformations in special relativity, in which the time in a moving reference frame depends linearly on time and space in another reference frame.  Here, we find some range of possible $a$ for which the time function is causal, reminiscent of the fact that in special relativity, only certain linear transformations will preserve timelike intervals.
In the same spirit we remark that the computer science tradition of treating space and time complexities separately
must at some point encounter special relativity.

Now we describe a more general form of ``borrowing".
Given a circuit $C$ implementing unitary $U$, with a causal time function $\tau$, let a ``borrowing function" $f(\cdot)$ be a function from sites to sites with the following properties.
First, the range of the borrowing function is a subset of the set of sites with a physical qudit.
Next,
for any site $i$, for each site $j$ such that $f(j)=i$, let $S_j(i)$ be the set of $(S,a)$
such that there is a gate $U_{S,a}$ with $j \in S$.  
Let also $S_i(i)$ 
be the set of $(S,a)$
such that there is a gate $U_{S,a}$ with $i \in S$.  
Let $T_j(i)$ be the image $S_j(i)$ under $\tau$, i.e., in words $T_j(i)$ for is the sets of times at which some gate acts on $j$.
Then, for a borrowing function we require that for all $i$, for any pair $j,k\neq i$ with $f(j)=f(k)=i$ with $j\neq k$, we have that
either $T_j(i)<T_k(i)$ or $T_k(i)<T_j(i)$ where the inequality $T_j(i)<T_k(i)$ means that every element of $T_j(i)$ is smaller than every element of $T_k(i)$.
In words, we require that either all gates with support on $j$ act before all gates with support on $k$ or vice-versa.
Further we require that for any $j\neq i$ that either $T_j(i)<T_i(i)$ or $T_j(i)>T_i(i)$.  In words, we require that all gates with support on $j$ act before all gates with support on $i$, or vice-versa.

From such a borrowing function, we can construct a new circuit $C'$ as follows: for every gate $U_{S,a}$ in $C$, we interchange every ancilla qubit on site $j$ in its support with the physical qubit on site $f(j)$.
That is, define $\sigma_{i,j}$ to swap the physical degree of freedom on site $i$ with the ancilla degree of freedom on site $j$ and 
define
\be
U'_{S,a}\rightarrow \Bigl(\prod_{j\in S}\sigma_{f(j),j}\Bigr) U_{S,a} \Bigl( \prod_{j\in S} \sigma_{f(j),j}\Bigr),
\ee
with $U'_{S,a}$ the gates in circuit $C'$.
The gates $U'_{S,a}$ do not act on the ancilla qubits, so we can then define in the obvious way a circuit $C''$ that acts on a system that includes only physical degrees of freedom, implementing unitary $U''$ with $U'=U'' \otimes I$.

We now show that
\begin{lemma}
If $U$ acts trivially on the ancilla degrees of freedom, then
$U=U'$.
\begin{proof}
Let there be $N$ sites.
Let $\Lambda_1,\ldots,\Lambda_N$ be a sequence of subsets of sites, with $|\Lambda_m|=m$ and each $\Lambda_i$ obtained from $\Lambda_{i-1}$ by a adding a single site.
We prove the lemma by considering a sequence of circuits $C_0,C_1,C_2,\ldots,C_N$ with $C=C_0$ and $C_N=C'$ and showing that $C_m$ and $C_{m-1}$ implement the same unitary for each $m>0$.

We define $C_m$ to have gates $U^{(m)}_{S,a}$ with
\be
U^{(m)}_{S,a}=\Bigl(\prod_{j\in S\cap \Lambda_m}\sigma_{f(j),j}\Bigr) U_{S,a} \Bigl( \prod_{j\in S\cap \Lambda_m} \sigma_{f(j),j}\Bigr).
\ee
In words, in each circuit $C_m$ we make the swap of ancilla qubit on site $j$ with physical qubit on site $f(j)$ only for those $j\in \Lambda_m$.

Let $C_m$ implement unitary $U_m$.  We will show that $U_m=U_{m-1}$.
Let $\Lambda_m \setminus \Lambda_{m-1}={j}$ with $i=f(j)$.
Let $t_{min}$ be the smallest element of $T_j(i)$ and $t_{max}$ be the largest element.

Define $V^{(m-1)}_b
=\prod_{(S,a) \; {\rm s. t.} \; \tau((S,a))=b} U^{(m-1)}_{S,a}$,
 as in lemma \ref{Vbdef}.
By assumption, the product $V^{(m-1)}_{t_{max}} \circ \ldots \circ V^{(m-1)}_{t_{min}}$ acts trivally on the ancilla qudit on site $j$, since all gates that act on this ancilla are in this product.  Further, by assumption, 
 the product $V^{(m-1)}_{t_{max}} \circ \ldots \circ V^{(m-1)}_{t_{min}}$ acts trivally on the physical qudit on site $i$
 since either $T_i(i)<T_j(i)$ or $T_i(i)>T_j(i)$.
 Hence, 
 \begin{eqnarray}
V^{(m-1)}_{t_{max}} \circ \ldots \circ V^{(m-1)}_{t_{min}} &=&
\sigma_{f(j),j)} V^{(m-1)}_{t_{max}} \circ \ldots \circ V^{(m-1)}_{t_{min}} \sigma_{f(j),j}
\\ \nonumber
&=& V^{(m)}_{t_{max}} \circ \ldots \circ V^{(m)}_{t_{min}}.
\end{eqnarray}
Hence $U_{m}=U_{m-1}$.
\end{proof}
\end{lemma}

Thus, we have shown that the existence of a causal time function and a borrowing function enables us to remove the ancillas.
Now we give some conditions under which, given some circuit, we can construct a causal time function and a borrowing function.
The goal is to remove ancillas without increasing the circuit depth or range of the gates by an excessive amount.

We make the following assumption on the control space and sites.  The assumption is phrased in terms of the number of physical and ancilla qudits, so they may be applied to the setting in which each site has one physical and one ancilla qudit or to a more general setting with a lower density of physical qudits.
\begin{assumption}
\label{a1}
There are constants $0<c\leq c'$ and $d>0$
and there is some $\ell_0$ which 
depends only on the range and depth of the quantum circuit and on $c,c',d$,
such that:
for any $\ell\leq \ell_0$, the number of physical qudits within distance $\ell$ of $x$ is at least
$\lfloor c\ell^{d} \rfloor$ and further the number of ancilla
qudits with distance from $x$ in the interval $[\ell-1,\ell]$ is at most
$\lfloor c'\ell^{d-1} \rfloor$.
\end{assumption}
Remark: note that the lower bound on the number of physical qudits is on the number contained in some ball of radius $\ell$ while the upper bound on the number of ancilla is on some ``shell" of thickness $1$.

It is not clear to us how much this assumption can be generalized.  While the assumption holds for many reasonable choices of control space and control map, it is likely that for other control spaces borrowing functions can still be constructed.  A particularly interesting case to consider would be a hyperbolic space with a uniform density of qudits (uniform on some scale) and with the range of the circuit large compared to the curvature radius.  Our construction given Assumption \ref{a1} will use a particularly simple choice of function $\hfn$ to construct a causal time function.  After giving this construction, we will discuss other possible choices of function $\hfn$ and explain how to construct a borrowing function on the hyperbolic plane.

We consider a quantum circuit whose gates have range $r$ and with the quantum circuit having depth $\dep$.
Recall that an $\epsilon$-net is a set of points
 $x_1,x_2,\ldots$, such that every
point in the space is within distance $\epsilon$ of some point in the $\epsilon$-net and no two distinct points in the sequence are distance less than $\epsilon$ of each other.
We construct a $\ell_0$-net for some $\ell_0$ that we pick later.
We choose the time function $\tau(S,a)=\lfloor {\rm dist}(S,B)\rfloor+ca$ with $c=r+1$ with $B$ being the set of points in the net.  By lemma \ref{tfun}, this time function is causal.

We now construct a borrowing function.
Recall that for a borrowing function, we require that if $f(j)=f(k)$ with $j\neq k$,
then either all gates with support on $j$ act before all gates with support on $k$ or vice-versa.
Consider a site $j$ with $\lfloor {\rm dist}(j,B)\rfloor=m$.  Then, gates acting on $j$ act at times in the interval $[m,m+(r+1)\dep]$, where we use the time function $\tau$.
We will decompose the set of ancilla qudits into several disjoint sets, called $S_0,S_1,\ldots$, where $S_a$ is the set of ancilla qudits $j$ with $(r+1)\dep a \leq \lfloor {\rm dist}(j,B)\rfloor < (r+1)\dep(a+1)$.
Then, gates acting on qudits in in $S_a$ will all act before (or will all act after) those acting on qudits in $S_b$ for any $a,b$ with $|a-b|>1$.
So, for qudit $j\in S_a$ and $k\in S_b$ with $|a-b|>1$, we do not have a constraint that $f(j)\neq f(k)$.  We now make a choice to ensure that
for $j\in S_a$ and $k\in S_{a+1}$ that $f(j) \neq f(k)$.

Let us first introduce some notation.
For each site $j$, let $b(j)$ be a closest point in the net $B$ and for any $x\in T$ let $T(x)$ be the set of $j$ with $b(j)=x$.
We will decompose $T(x)$ into two disjoint sets, $T_1(x),T_2(x)$, each of size at least $\lfloor T(x)/2 \rfloor$.
We will then make an arbitrary choice: we will choose that for $j\in T(x)$, if $j\in S_a$ for $a$ even (or odd), we will pick $f(j)$ in $T_1(x)$ (respectively, $f(j)\in T_2(x))$.
This guarantees the needed property to get a borrowing function.
Further, by construction ${\rm dist}(j,f(j))\leq 2\ell_0$.

So, it remains to show that we can indeed make such a choice of $f(j)$.  However, this is possible by the density assumption:
for each point $x\in B$, $T(x) \cap S_a$ is bounded by $c'((r+1)d)^d(a+1)^{d-1}$.  However, the number of
physical qudits closer to $x$ than to any other point in the net is at least $c \ell_0^d$.
For sufficiently large $\ell_0$, the second number is larger than the first and so the matching exists.

Thus we have
\begin{theorem}
Under Assumption \ref{a1} above, for sufficiently large $\ell_0$ depending only on $c,c',d$ and on the range and depth of the quantum circuit,
there is a borrowing function with ${\rm dist}(j,f(j))\leq \ell_0$ and so ancillas may be removed.
\end{theorem}

As a corollary
\begin{corollary}
Consider any pair of QCA $\alpha,\beta$ with range $R$.
Let the Assumption \ref{a1} hold. Further, we assume that there are {\it no} ancilla qudits.
Then $\alpha \circ \beta \circ \alpha^{-1} \circ \beta^{-1}$ is a quantum circuit with depth and range depending only on $R,c,c',d$.
\begin{proof}
Using ancillas, $\alpha \circ \beta \circ \alpha^{-1} \circ \beta^{-1}$  is a quantum circuit.  Since assumption \ref{a1} held without ancillas, when we add one extra ancilla for each physical qudit it still holds and so we may remove the ancillas.
\end{proof}
\end{corollary}

Now let us consider to what extent we can generalize Assumption \ref{a1}.  First, we emphasize that {\it some} kind of lower bound on the density of physical qudits is necessary.  To see this, consider the following simple example.  Imagine a large square lattice, size $\ell\times \ell$ for some $\ell \gg 1$, with periodic boundary conditions in both directions and with one qudit per site.  Let the qudits with horizontal coordinate equal to $0$ or to $\ell/2$ be physical qudits, and let all other qudits be ancilla qudits.  Then, there is a quantum circuit with range $O(1)$ which implements a vertical shift by $+1$  on one line of physical qudits (say the qudits with horizontal coordinate $0$) and implements a vertical shift by $-1$ on the other line of physical qudits (those with horizontal coordinate $\ell/2$.  This circuit is implemented by a ``swindle" (see for example discussion of such a circuit in Ref.~\cite{fh}).  On the other hand, no such circuit exists which acts just on the physical qudits.

However, we still expect that some form of ancilla removal can be implemented on more general control spaces so long as the coarse density of physical qudits is comparable to the coarse density of ancilla qudits.  Certainly, if there is one physical and one ancilla qudit per site, one expects that ancilla removal can be implemented more generally.
It seems that to do this, one way is to generalize the choice of function $\hfn$ in lemma \ref{tfun}.
For example, :
\begin{lemma}
Let $\hfn: X\ra [0,T]$ be $1$-Lipschitz where $T$ is $O(1)$.  Assume that there is one physical and one ancilla qudit per site.
Consider a quantum circuit of depth $\dep$ with gates of range $r$.
Assume that there is an involution $f$ on the set of sites such that
for any site $x$, $|\hfn(x)-\hfn(f(x))|>\dep(r+1)$ and  ${\rm dist}(x-f(x))\leq r \cdot O(1)$.
In words, this involution maps each site to a nearby site with a sufficiently different value of $\hfn$.

Then, 
$\tau(S,a)=\lfloor \hfn(S) \rfloor+(r+1) a$ is a causal time function and
$f(\cdot)$ defines a borrowing function for this time function.

Removing ancillas using this time function and borrowing function increaes the depth and range of the quantum circuit by an $O(1)$ multiplicative factor.
\begin{proof}
Consider a pair $x,f(x)$.  Assume without loss of generality $\hfn(x)>\hfn(f(x))+d(r+1)$.
Since the quantum circuit $C$ has depth $\dep$, all gates in circuit $C'$ acting on the qudits on site $f(x)$ are executed at time at most $$\lfloor \hfn(f(x)) \rfloor + \dep(r+1)\leq \hfn(f(x)) + \dep(r+1)<\hfn(x).$$  All gates in $C'$ acting on the qudits on site $x$ are in a gate which is supported on some set $S$ within distance $r$ of $x$ so they are executed at time at least
executed at time at least $$\lfloor \hfn(S) \rfloor + (r+1)\geq \hfn(S)+r \geq \hfn(x).$$
So, the borrowing condition is satisfied.
\end{proof}
\end{lemma}

Note that we want $\hfn$ to have a bounded variation, so that the depth of the circuit does not increase too much.
However, to find the involution $f$, it is important that the density
 (on some coarse scale) of sites in the inverse image under $\hfn$ of various intervals $[i,i+1]$ is roughly constant.  That is, if we try the choice $\hfn(x)={\rm dist}(X,B)$ for a hyperbolic space, we find that most of the density is in the inverse image of the interval $[\ell_0-1,\ell_0]$, i.e., most points are far from the net, so we cannot find a physical qudit to borrow to remove those ancillas.

For the two dimensional hyperbolic plane, however, we can construct borrowing function as follows.  Choose a periodic tiling by large (compared to the circuit range) tiles.  Choose the height function to be $0$ at the center of the tiles and increase radially outward.  For example, far enough from the boundary of a tile, choose the height function to simply be the distance from the center of a tile.  Assuming a site density which is uniform on this scale with one physical and one ancilla qudit per site, we can remove the ancillas far enough from the boundary of the tiles simply because most of the volume is not near the center of the tiles.  What is left is then some quantum circuit supported near the $1$-skeleton of this tiling.  We can then find some borrowing to remove ancillas for this circuit, choosing the height function to increase from the center of the edges to their boundary.  We leave a precise construction and a generalization to higher dimensions for future work.

A regular tiling of hyperbolic space is generally constructed by selecting a co-compact reflection group $\Gamma$ and then each tile is a convex fundamental domain.  Surprisingly there is  Vinberg's theorem\cite{vinberg1967discrete} which states that above dimension $30$ there are no cocompact reflection groups acting on hyperbolic space ( 30 may not be a sharp bound).  This rather surprising theorem is number  theoretic in nature and the does not yet seem to have a geometric explanation. But it serves as a warning  that it may not be straight forward to generalize the sketch we presented for ${\mathbb H}^2$ to hyperbolic space in all higher dimensions.

\section{Coherent Families and the official definition of QCA}
\label{sec:cf}

This section completes the circle.
We began with naive definitions of QCA and quantum circuits and then enhanced them to sequence-based definitions.
In this section,
we prove coherence theorems, Theorem \ref{mpleating} and \ref{deform}, to show that an initial QCA $\alpha$ (of sufficiently small range $R$) in fact contains all the necessary information in the first place. The ``mother" $\alpha_0 \in$ QCA of a {\it coherent} sequence $\lbar{\alpha} \in$ QCA effectively determines $\lbar{\alpha}$ which obeys a strong uniqueness property. This allows us to dispense with sequences and return, now with greater confidence, to the naive definition.

This section will focus on path coherent families and equivalence up to paths.  The next section \ref{circ} will discuss circuit equivalance of families.

Consider a coherent family $\lbar{\alpha} = \{\alpha_0, \alpha_1, \dots\}$. Each $\alpha_i$ is of range $\frac{R}{2^i}$, with $\alpha_0 = \alpha$ called the \emph{mother}. 
It is required in definition \ref{cohfam} that $\alpha_{i+1}$ be obtained from $\alpha_i$, $i \geq 0$ by a path of QCA.

However, for purposes of applying the ``pleating" construction in this section, we will
allow $\alpha_{i+1}$ to be obtained from $\alpha_{i}$ by what we will call
a \emph{deformation}. We will see later that this in fact is not a change in the definition: all of the allowed deformations
can in fact be described by stabilization and quantum circuit and hence our construction will give a circuit coherent family, not just a path coherent family.

We will see that the pleating construction will give a circuit coherent family.  The uniqueness results in this section will apply however to path coherent families more generally and will deal with a more general notion of equivalence under paths.  In section \ref{circ} we will deal with circuit coherent families and equivalence using circuits.

A deformation is a composition of $c_1$ local deformations, $c_1$ a constant independent of $i$ depending only on the control space $X$. (We may take $c_1 = \mathrm{dim}(X) + 1$.) A \emph{local deformation} consists of a composition of four operations all supported on the set $h_j \subset X$ where $h_j$ is the disjoint union of well seperated (contractible) balls. In our application $X$ is given a handle decomposition $h$ consisting of 0-handles $h_0$, 1-handles $h_1, \dots, d$-handles $h_d$. The well-separated condition can be taken to be that the distance between any two disjoint handles in $X$ is $>10 \cdot 2^d R_i$, $d$ the dimension of the control space $X$, and each handle will be assumed to have diameter bounded by a constant times $R_i$; a precise sufficient condition is given below (If $X$ is P.L. triangulated, a well-separated handle decomposition can be obtained from a refinement of the original triangulation). The four operations, which together constitute what is allowed by a quantum in a stable setting, are:
\begin{equation}\label{operations}
\begin{split}
	& \text{(1) Stabilization: the index set $I$ is enlarged to $I^\pr$ with new finite Hilbert spaces} \\[-0.25em]
	& \hspace{1.5em}\text{$\mathcal{H}_i$, $i \in I^\pr \setminus I$ introduced and the control map $\mathcal{I}: I \ra X$ extended by a locally} \\[-0.25em]
	& \hspace{1.5em} \text{finite mapping of $I^\pr \setminus I$ to $X$.} \\
	& \text{(2) Acting by arbitrary quantum circuits over $h_i$ .} \\[-0.25em]
	& \hspace{1.5em} \text{independent of $i$.} \\
	& \text{(3) Destabilization: the inverse operation to (1)} \\
	& \text{(4) Composing the control map $\mathcal{I}: I \ra X$ by a diffeomorphism $g: X \ra X$,} \\[-0.25em]
	& \hspace{1.5em} \text{isotopic to the identity. }
\end{split}
\end{equation}
Note: It is possible for a deformation to implement all stabilizations at one time at the beginning and all destabilizations at one time at the end. Then since steps (2) and (4) can be done adiabatically, it is legitimate to think of ``deformation" indeed as a continuous process.

Further, clearly, (1-3) are all examples of quantum circuits or stabilization (the destabilization operation (3) is equivalent to stabilizing $\alpha_{i+1}$), while (4) can be described by circuits and stabilization by tensoring with ancillas and using swap gates.  For example, to a site from some point $x$ to some point $x'$, we tensor with an additional ancilla at $x'$ (stabilization), conjugate by a swap, and then remove the the original degree of freedom at $x$ (destabilization).
Since the handle diameter is bounded by $R_i$, (2) can be implemented by a single gate of range bounded by $R_i$ and hence is a quantum circuit.  Similarly, the range bound on the swaps in (4) is obeyed.

Now we define the group $\lbar{\text{QCA}(X)}$ to be the group of coherent families over $X$ subject to the equivalence relation $\lbar{\alpha} \equiv \lbar{\beta}$ if there exists a $k_0$ so that for $k \geq k_0$, $\alpha_k = \beta_k$ i.e. sequences are equivalent if they agree on a terminal segment. Notice that all finite compositions are defined and (up to equivalence) act with arbitrarily small range.

We also consider finite depth quantum circuits fdqc$(X)$ on $X$ whose depth is $O(1)$ and \emph{not} allowed to diverge as the degrees of freedom become more closely packed in $X$ and the sequence index $i$ approaches infinity. When $\lbar{\text{fdqc}(X)}$ is similarly defined as terminal equivalence classes of families of circuits, with each family having range of the gates in the circuit bounded by $R/2^i$ for some $i$, we have
\begin{equation}
	\lbar{\text{fdqc}}(X) \lhd \lbar{\text{QCA}}(X)
\end{equation}
and as proved in \cites{fh,hfh} the quotient is again abelian:
\begin{equation}
	\lbar{Q}(X) = \lbar{\text{QCA}}(X) / \lbar{\text{fdqc}}(X)
\end{equation}

Theorem \ref{deform} will allow us to complete the circle and remove the ``bars" which we added for mathematical precision. It says that if the initial $R$ is small w.r.t. the handle structure on $X$, then $R$-sequences exist and are unique in a very strong sense. But first we state a purely topological result.

Given a $d$-dimensional Riemannian manifold $(X, g)$, possibly with boundary and possibly non-compact, define a map $p: X \ra X$ to be a \emph{pleating} if it is a finite (at most $d+1$-fold) composition of \emph{basic pleatings} which are defined to be smooth maps $p_i: X \ra X$ so that for every point in the source smooth local coordinates in source and target can be found so that the map takes on one of three forms:

\begin{enumerate}
	\item rank $=d$, in local coordinates about $\vec{x} = (0, \dots, 0)$, $p(x_1, \dots, x_d) = (x_1, \dots, x_d)$
	\item fold, $p(x_1, \dots, x_d) = (x_1^2, x_2, \dots, x_d)$, rank $=d-1$
	\item cusp, $p(x_1, \dots, x_d) = (x_1^3 + x_2x_1, x_2, \dots, x_d)$, rank $=d-1$
\end{enumerate}

Thus $p = p_n \circ \cdots \circ p_1$.

Given $p$, let $X_p$ be the graph of $p$; $X_p = \{(x, p(x)\} \subset X \times X$. Of course $\pi_1: X_p \ra X$ is a diffeomorphism. Our interest in $X_p$ comes from the fact that it inherits a piecewise-smooth Riemannian metric $g_p$ from the second projection: $g_p = \pi_2^\ast(g)$. The next theorem says that we may choose $p$ so that the first projection $\pi_1: X_p \ra X$ is coarsely compressive.

\begin{theorem}[Manifold pleating]\label{mpleating}
	Let $(X_g)$ be a Riemannian manifold of dimension $d \geq 1$, possibly with boundary and possibly non-compact. For every $\epsilon$, $\epsilon^\pr > 0$ there is a pleating $p: X \ra X$ so that $\mathrm{dist}_g(x_1, x_2) \leq \epsilon \mathrm{dist}_{g_p}((x_1, p(x_1)), (x_2, p(x_2)))$, whenever $\mathrm{dist}_g(x_1, x_2) \geq \epsilon^\pr$.

	The next theorem expresses the fact that the set of pleatings is a very well-behaved function space, and well-adapted for application to parameter families.
\end{theorem}

\begin{theorem}\label{functionspace}
	The space of pleatings homotopic to the identity $\mathrm{id}_X$ (the topology is either the compact open topology w.r.t. $C^0$ or any $C^k$) is both contractible and locally contractible. Pleatings may be constructed pointwise near the identity and also relative to the identity on any open subset of $X$.
\end{theorem}

\begin{proof}
	The allowed singularity types 1) immersion 2) fold and 3) cusp are sufficiently flexible that the space of these maps satisfies Gromov's $h$-principle \cites{gromov1986,eliashberg2011} without dimension restriction (compare this to the case of immersions alone in which according to Smale-Hirsch theory \cite{hirsch1959} the index of handles must be strictly less than the ambient dimension of the immersion.) As a result the space of such maps has the same weak homotopy type as the corresponding function space: $\{ \text{maps } X \ra X \text{ pointwise close to the identity}\}$ which is clearly contractible and locally contractible.
\end{proof}

Before giving the proof of Theorem \ref{mpleating}, which inducts on a handle decomposition, it is good to see how it works for $X$ an interval, where no handle decomposition is needed. In this case $p$ can be taken to have the graph indicated in Fig.~\ref{fig:corepleat}.

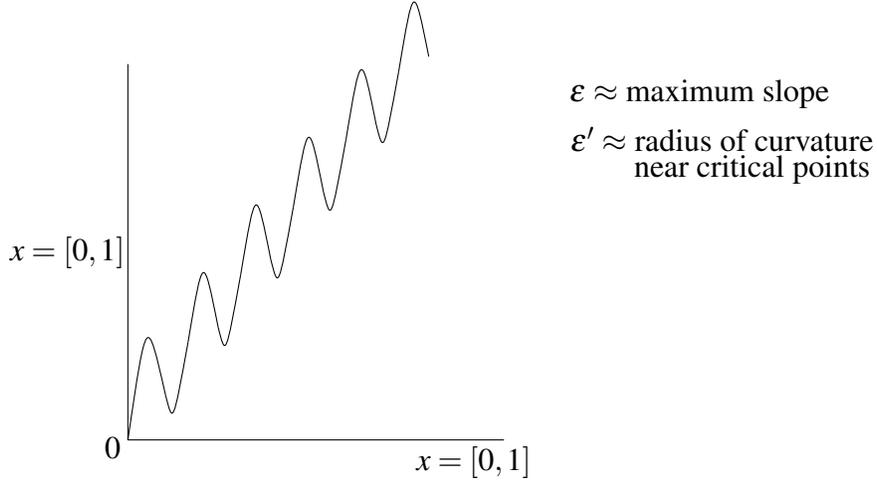
\begin{figure}[!h]
	\centering
	\begin{tikzpicture}
		\draw (0,0) -- (5,0);
		\draw (0,0) -- (0,5);

		\node at (-0.2,-0.1) {0};
		\node at (4.6,-0.3) {$x = [0,1]$};
		\node at (-0.8,2.5) {$x = [0,1]$};

		\draw (0,0) .. controls (0.25,1.7) .. (0.5,0.6);
		\draw (0.5,0.6) .. controls (0.6,0.2) .. (0.8,1.3);
		\draw (0.8,1.3) .. controls (1,2.5) .. (1.2,1.5);
		\draw (1.2, 1.5) .. controls (1.3, 1.1) .. (1.5, 2.2);
		\draw (1.5, 2.2) .. controls (1.7, 3.4) .. (1.9, 2.4);
		\draw (1.9, 2.4) .. controls (2, 2) .. (2.2, 3.1);
		\draw (2.2, 3.1) .. controls (2.4, 4.3) .. (2.6, 3.3);
		\draw (2.6, 3.3) .. controls (2.7, 2.9) .. (2.9, 4);
		\draw (2.9, 4) .. controls (3.1, 5.2) .. (3.3, 4.2);
		\draw (3.3, 4.2) .. controls (3.4, 3.8) .. (3.6, 4.9);
		\draw (3.6, 4.9) .. controls (3.8, 6.1) .. (4, 5.1);

		\node at (7.6,4.6) {$\epsilon \approx$ maximum slope};
		\node at (7.9,4) {$\epsilon^\prime \approx$ radius of curvature};
		\node at (8.3,3.6) {near critical points};
	\end{tikzpicture}
	\caption{Graph $p$, a basic 1D pleating.}
	\label{fig:corepleat}
\end{figure}

Remark: In applications it is useful that the pleats can be chosen large, $\gg R$, so when a QCA is pulled back under $(p_i)^{-1}$ we see an isomorphic copy restricted to a large open set of $X_p$. We will not make this completely explicit (it is more tedious than illuminating) but merely note, as we set up the handle by handle pleating, that the size of the handles should be large relative to the QCA radius $R$. The constants we give are comfortable overestimates. We choose a fine handle decomposition of $X$, sufficiently coarse w.r.t. the QCA radius $R$ so that the distance between disjoint handles $\geq 10 \cd  2^d R$ and handles are metrically convex with diameters in the range $100 \cd 2^dR$ to $1000 \cd 2^dR$.

The reason for the $2^d$ factors is that while we increase control in the co-core directions of a handle, $h_i$  of index $i$, by compressing these directions, this comes at the cost of stretching, say by a factor of 2,  parts of the higher index $j$, $j>i$, handles which attach to $h_i$. It is important that the cumulative effect of these stretchings does not increase the range $R$ of the QCA so much that the handle structure fails to be (approximately) preserved.  The key is point is that at any time during the handle by handle modifications of the QCA its action on operators supported near the central region of any handle core be mapped to operators supported well within that handle.

\begin{proof}[Proof sketch of Theorem \ref{mpleating}]
	The proof is visual and many pictures are now supplied. The first step is to shrink distances with $h_0$, the zero handles. This is done by radially compressing each 0-handle toward its origin, see Fig.~\ref{figpyr}.

	\begin{figure}[!ht]
		\centering
		\begin{tikzpicture}[scale=0.9]
			\draw (0,0) circle (2);
			\draw (0,0) circle (0.75);
			\draw (8,0) circle (2);
			\draw (8,0) circle (0.75);

			\draw [dashed] (0.65, 0.4) -- (1.5, 1.3) -- (6.5, 1.3) -- (7.35, 0.4);
			\draw [dashed] (0.65, -0.4) -- (1.5, -1.3) -- (6.5, -1.3) -- (7.35, -0.4);
			\draw [dashed] (-0.65, 0.4) -- (-1.5, 1.3) -- (-3.5, 1.3);
			\draw [dashed] (-0.65, -0.4) -- (-1.5, -1.3) -- (-3.5, -1.3);

			\draw [dashed] (8.45,0.6) -- (9.25, 1.55) -- (9.6, 2.7);
			\draw [dashed] (8.45,-0.6) -- (9.25, -1.55) -- (9.6, -2.7);
			\draw [dashed] (8.75,0) -- (10,0) -- (11.2,-0.9);

			\node at (4,1.6) {2-handle};
			\node at (4,0) {1-handle};
			\node at (4,-1.6) {2-handle};

			\draw [->] (0.7,-1.6) -- (0.4, -0.8);
			\draw [->] (0.7,1.6) -- (0.4, 0.8);
			\node at (-0.4,-1.2) {0-handle};
			\draw [->] (7.2, 1.45) -- (7.6, 0.75);
			\draw [->] (7.2, -1.45) -- (7.6, -0.75);
			\node at (8.2,-1.4) {0-handle};
			\draw [->] (9.4, 1) -- (8.7, 0.4);
		\end{tikzpicture}
		\caption{}
		\label{figpyr}
	\end{figure}

	This at first seems a pyrrhic victory because the 1-handles must correspondingly stretch. However, the 1-handles are now pleated along the 1D core of the handle and radially compressed in the $(d-1)$-cocore direction. In co-core directions where pleating does not shrink distances, we rely on radial compression as in the 0-handle case. Pleating means the differential becomes rank $d-1$ along a sequence of sphere $S^{d-1}$ embedded transverse to the core of the one handle. There are cusp singularities along an equatorial $S^{d-2}$ subspheres. In the $d=2$ case these are $(S^1, S^0)$ pairs; the map is indicated in Fig.~\ref{fig:pleats}, and Fig.~\ref{fig:corepleat} illustrates the pleating along the core in the 1D case, e.g. $X \cong S^1$.

	\begin{figure}[!h]
		\centering
		\begin{tikzpicture}[scale=1.1]
			\draw (-3,-1) -- (-3,1);
			\draw (3,-1) -- (3,1);
			\draw (-3, 1) -- (3,1);
			\draw (-3, -1) -- (3,-1);
			\draw (-2, 0) ellipse (0.3 and 0.75);
			\draw (-1, 0) ellipse (0.3 and 0.75);
			\draw (0, 0) ellipse (0.3 and 0.75);
			\draw (1, 0) ellipse (0.3 and 0.75);
			\draw (2, 0) ellipse (0.3 and 0.75);
			\node at (2,0.75) [circle,fill,inner sep=1pt]{};
			\node at (2,-0.75) [circle,fill,inner sep=1pt]{};
			\node at (1,0.75) [circle,fill,inner sep=1pt]{};
			\node at (1,-0.75) [circle,fill,inner sep=1pt]{};
			\node at (0,0.75) [circle,fill,inner sep=1pt]{};
			\node at (0,-0.75) [circle,fill,inner sep=1pt]{};
			\node at (-1,0.75) [circle,fill,inner sep=1pt]{};
			\node at (-1,-0.75) [circle,fill,inner sep=1pt]{};
			\node at (-2,0.75) [circle,fill,inner sep=1pt]{};
			\node at (-2,-0.75) [circle,fill,inner sep=1pt]{};

			\draw (-3, 0) -- (2.6, 0);
			\draw [->] (-0.5, 1.5) -- (-0.5, 0.2);
			\node at (-0.5,1.7) {core};

			\draw [->] (0, -1.5) arc (-90:10:3.3 and 0.6);
			\draw [->] (0, -1.5) arc (-90:-190:3.3 and 0.6);

			\node at (4.2,0) {1-handle};
			\node at (0,-1.8) {attaching regions};
			\node at (0,-2.4) {(a) arcs are folds, dots are cusps pre-image};
			\node[opacity=0] at (-4.2,0) {1-handle};
		\end{tikzpicture}

		\vspace{2em}

		\begin{tikzpicture}[scale=1.3]
			\draw (-3, -5) -- (3, -5);
			\draw (-3,-4.3) to [out=0, in=150] (-1.91,-3.7) to [out=-60, in=45] (-2.1,-3.98) to [out=225, in=45] (-2.25,-4.12) to [out=225, in=180] (-2.15,-4.3);
			\draw (-2.15,-4.3) to [out=0, in=150] (-1.06,-3.7) to [out=-60, in=45] (-1.25,-3.98) to [out=225, in=45] (-1.4,-4.12) to [out=225, in=180] (-1.3,-4.3);
			\draw (-1.3,-4.3) to [out=0, in=150] (-0.21,-3.7) to [out=-60, in=45] (-0.4,-3.98) to [out=225, in=45] (-0.55,-4.12) to [out=225, in=180] (-0.45,-4.3);
			\draw (-0.45,-4.3) to [out=0, in=150] (0.64,-3.7) to [out=-60, in=45] (0.45,-3.98) to [out=225, in=45] (0.3,-4.12) to [out=225, in=180] (0.4,-4.3);
			\draw (0.4,-4.3) to [out=0, in=150] (1.49,-3.7) to [out=-60, in=45] (1.3,-3.98) to [out=225, in=45] (1.15,-4.12) to [out=225, in=180] (1.25,-4.3);
			\draw (1.25,-4.3) to [out=0, in=150] (2.34,-3.7) to [out=-60, in=45] (2.15,-3.98) to [out=225, in=45] (2,-4.12) to [out=225, in=180] (2.1,-4.3);
			\draw (2.1,-4.3) to [out=0, in=150] (3.19,-3.7) to [out=-60, in=45] (3,-3.98) to [out=225, in=45] (2.85,-4.12) to [out=225, in=180] (2.95,-4.3);

			\draw [->] (0, -4.4) -- (0, -4.8);

			\node at (0,-5.6) {(b) This projection indicates $p$ restricted to};
			\node at (0.4,-6) {the core, the folds are gradually reduced};
			\node at (0.3,-6.4) {to nothing as one moves parallel to the};
			\node at (-1,-6.8) {core on either side};
		\end{tikzpicture}
		\caption{}
		\label{fig:pleats}
	\end{figure}

	Next, the 2-handles are pleated. When $d=2$ this has a geometry reminiscent of a bellows, see Fig.~\ref{figbellows}.

	\begin{figure}[!ht]
		\centering
		\begin{tikzpicture}
			\draw (0,0) circle (1);
			\draw (0,0) circle (1.25);
			\draw (0,0) circle (1.5);
			\draw (0,0) circle (1.7);
			\draw (0,0) circle (1.875);
			\draw (0,0) circle (2);
			\draw (0,0) circle (2.1);

			\node at (-4,0.9) {pairs of fold circles};
			\node at (-4,0.5) {(no cusps when $d=2$)};
			\draw [->] (-2.2, 0.9) -- (-1.65, 0.7);
			\node at (4,1.2) {attaching region};
			\draw [->] (2.6, 1.2) -- (1.9,1);
			\node at (3.6,0) {2-handle};
			\node at (0,-2.5) {(a)};
			\node[opacity=0] at (5.9,0) {A};
		\end{tikzpicture}

		\vspace{2em}

		\begin{tikzpicture}
			\draw (-4,0) -- (4,0);
			\draw (0,0) -- (0,3);
			\draw (0,0.5) -- (0.1,0.6);
			\draw (0,0.5) -- (-0.1,0.6);
			\draw [->] (-0.5,2.9) arc (180:360:0.5 and 0.35);
			\node at (1.2,2.75) {rotate};

			\draw (-2,2) -- (2,2);
			\draw (-2,2) arc (90:270:0.125);
			\draw (-2,1.75) -- (-0.4,1.75);
			\draw (-0.4,1.75) arc (90:-90:0.125);
			\draw (-0.4,1.5) -- (-2.4,1.5);
			\draw (-2.4,1.5) arc (90:270:0.125);
			\draw (-2.4,1.25) -- (-0.8,1.25);
			\draw (-0.8,1.25) arc (90:-90:0.125);
			\draw (-0.8,1) -- (-2.8,1);
			\draw (-2.8,1) arc (90:270:0.125);
			\draw (-2.8,0.75) -- (-1.2,0.75);
			\draw (-1.2,0.75) arc (90:-90:0.125);
			\draw (-1.2,0.5) -- (-3.4,0.5);

			\draw (2,2) arc (90:-90:0.125);
			\draw (2,1.75) -- (0.4,1.75);
			\draw (0.4,1.75) arc (90:270:0.125);
			\draw (0.4,1.5) -- (2.4,1.5);
			\draw (2.4,1.5) arc (90:-90:0.125);
			\draw (2.4,1.25) -- (0.8,1.25);
			\draw (0.8,1.25) arc (90:270:0.125);
			\draw (0.8,1) -- (2.8,1);
			\draw (2.8,1) arc (90:-90:0.125);
			\draw (2.8,0.75) -- (1.2,0.75);
			\draw (1.2,0.75) arc (90:270:0.125);
			\draw (1.2,0.5) -- (3.4,0.5);

			\node at (0,-0.5) {(b) Illustration should be rotated around the $z$-axis};
			\node at (0.35,-0.9) {and then vertically projected to pleat a 2-handle};
		\end{tikzpicture}
		\caption{}
		\label{figbellows}
	\end{figure}

	For a 3D 2-handle the pairs of circle folds become torus folds with two longitudinal curves on each torus becoming the source of cusp singularities, see Fig.~\ref{fig3d}.

	\begin{figure}[!ht]
		\centering
		\begin{tikzpicture}[scale=1.1]
			\draw (-3,1) -- (-3,-1);
			\draw (3,1) -- (3,-1);

			\draw (3,-1) arc (0:-180:3 and 0.35);
			\draw (0,1) ellipse (3 and 0.3);

			\draw (-2.35, -0.2) ellipse (0.05 and 0.6);
			\draw (2.35, -0.2) ellipse (0.05 and 0.6);
			\draw (2.35, -0.8) arc (0:-180:2.35 and 0.25);
			\draw (2.35, 0.4) arc (0:-180:2.35 and 0.2);
			\draw [dashed] (2.35, 0.4) arc (0:180:2.35 and 0.2);

			\draw (1.7,-0.3) ellipse (0.05 and 0.35);
			\draw (-1.7,-0.3) ellipse (0.05 and 0.35);
			\draw (-1.7, -0.65) arc (180:360:1.7 and 0.15);
			\draw (1.7, 0.05) arc (0:-180: 1.7 and 0.15);
			\draw [dashed] (1.7, 0.05) arc (0:180:1.7 and 0.1);

			\draw (1.1, -0.4) ellipse (0.05 and 0.2);
			\draw (-1.1, -0.4) ellipse (0.05 and 0.2);
			\draw (1.1, -0.2) arc (0:-180:1.1 and 0.1);
			\draw [dashed] (1.1, -0.2) arc (0:180:1.1 and 0.05);
			\draw (1.1, -0.6) arc (0:-180:1.1 and 0.1);

			\node at (4.82,0.5) {3D 2-handle};
			\node at (4.7,0) {attaching 1};
			\node at (5.02,-0.4) {region = $S^\prime \times I$};
			\draw [->] (3.8,0) -- (3.1, -0.1);
			\node[opacity=0] at (-5.5,0) {region};
		\end{tikzpicture}
		\caption{}
		\label{fig3d}
	\end{figure}

	The general prescription for a $d$-dimensional $k$-handle, i.e. the pair $(D^k \times D^{d-k}, \de D^k \times D^{d-k})$ is to pleat in the core $(D^k)$ direction with (fold locus, cusp locus) a sequence of embedded pairs
	\begin{equation}
		\begin{split}
			(S^{d-1}, S^{d-2}), & \text{\quad for } k=1 \\
			(S^{k-1} \times S^{d-k}, S^{k-1} \times S^{d-k-1}), & \text{\quad for } k\geq 2
		\end{split}
	\end{equation}
	and compress radially in the cocore $(D^{d-k})$ direction.
\end{proof}

Theorems \ref{mpleating} and \ref{functionspace} guide the stabilizations necessary to prove:

\begin{theorem}\label{deform}
	Given $\alpha = \alpha_0$, an $R$-QCA on a $d$-dimensional Riemannian manifold $X$, for sufficiently small $R$ there is a deformation of $\alpha_0$ through QCA of range $\leq 2R$ to $\alpha_1$ with range $\leq \frac{R}{2}$ (or any other specific fraction such as $\frac{R}{n}$ of $R$). Similarly, if $\alpha_0^t$ is a $k$-parameter family, $t \in D^k$, of QCA with range $\leq R$ for all $t$ and range $\leq \frac{R}{n}$ for $t \in \de D^k$, then there is a deformation continuous in the entire $t$-family and the identity near $\de D^k$ which reduces the range to $\frac{2R}{n}$ over the entire family. If we do not require the deformation to be the identity near $\de D^k$ the range can be reduced to $\frac{R}{m}$ over the entire family for any $m$.
\end{theorem}

\begin{proof}[Proof sketch]
	The key idea is that all pleats discussed in Theorem \ref{mpleating} and its proof can be simulated by introducing, via stabilization, an even $2j$ number of parallel sheets of $X$, with a copy of $\alpha$ on odd numbered copies and a copy of $\alpha^{-1}$ on even numbered copies. This is accomplished by writing $\alpha \circ \alpha^{-1} = \mathrm{id}$ on the odd sheets and simply $\mathrm{id}$ on the even sheets (stabilization) and then using the technique of Fig.~\ref{figpleat} to shift $\alpha$ to the even sheets. Finally the precise combinatorics of any pleating can then be simulated by judiciously returning bits of $\alpha$ to the odd sheets where there should \emph{not} be pleating. Finally cancel $\alpha$ with $\alpha^{-1}$ and destabilize. 
	
	Cancellation relies on the fact that $\alpha \otimes \alpha^{-1}$ is a quantum circuit\cites{arrighi2011unitarity,kcomm} of depth $O(1)$ and gate range bounded by $R$.  This circuit can be truncated	at the folds of the pleats, giving $\alpha\otimes \alpha^{-1}$ on a pair of pleats and $\Id$ outside the pleats (and some circuit near the fold; the details do not matter). This point is explained by Theorem 2.3 and its proof in \cite{fh}, and is illustrated in Fig.~\ref{figpleat} in the case of a single 1D zigzag.

	\begin{figure}[!ht]
		\centering
		\begin{tikzpicture}
			\draw (-4, 4) -- (4,4);
			\node at (0,4.2) {$\alpha$};
			\node at (6.1,4) {$X = \mathbb{R}$ in illustration};

			\draw [->] (0, 3.8) -- (0, 3.2);
			\draw (-4, 3) -- (4,3);
			\draw (-4, 2.7) -- (4,2.7);
			\draw (-4, 2.4) -- (4,2.4);
			\node at (-0.25,2.9) {$\alpha \circ \alpha^{-1}$};
			\node at (-3.6,2.52) {$\alpha$};

			\draw [->] (0,2.3) -- (0,1.7);
			\draw (-4,1.5) -- (4,1.5);
			\draw (-4,1.2) -- (4,1.2);
			\draw (-4,0.9) -- (4,0.9);
			\node at (-0.28,1.38) {$\alpha^{-1}$};
			\node at (-3.6,1.02) {$\alpha$};
			\node at (-0.5,1.65) {$\alpha$};
			\node at (6.1,3.05) {2nd extra sheet of $X$};
			\node at (6.1,2.65) {1st extra sheet of $X$};

			\draw [->] (0,0.7) -- (0, 0.1);
			\draw (-2, -0.2) -- (4, -0.2);
			\node at (-0.5,-0.05) {$\alpha$};
			\draw (-2,-0.2) -- (-3, -0.5) -- (2, -0.5) -- (1, -0.8) -- (-4, -0.8);
			\draw [dashed] (-4, -0.5) -- (-3, -0.5);
			\draw [dashed] (2, -0.5) -- (4, -0.5);
			\node at (-0.28,-0.32) {$\alpha^{-1}$};
			\node at (-3.6,-0.67) {$\alpha$};

			\node[opacity=0] at (-7.5,1) {g};
		\end{tikzpicture}
		\caption{}
		\label{figpleat}
	\end{figure}

	The parametric statements in Theorem \ref{deform} follow from the parametric properties of pleating described in Theorem \ref{functionspace}.
\end{proof}

Theorem \ref{deform} allows us both to implement the sequence definitions of $\lbar{\text{QCA}}$ and also to see that it is redundant. First given $\alpha$, a $R$-QCA on $X$, we can use Theorem \ref{deform} to construct a sequence $\{\lbar{\alpha} = \alpha_0, \alpha_1, \dots \} \in \lbar{\text{QCA}}(X)$. But now suppose $\lbar{\alpha}$ is the mother of another sequence $\{\alpha_0, \alpha_1^\pr, \dots \} \in \lbar{\text{QCA}}(X)$. There is a continuous path of $\frac{R}{2^n}$-QCA, $\alpha^t_n$, joining $\alpha_n$ to $\alpha_n^\pr$, $\alpha_n^0 = \alpha_n$ and $\alpha_n^1 = \alpha_n^\pr$ obtained by concatenating the deformation of $\alpha_0$ to $\alpha_n$ with the inverse of the deformation from $\alpha_0$ to $\alpha_n^\pr$. Theorem $\ref{deform}$ allows us to deform this path to a deformation from $\alpha_n$ to $\alpha_n^\pr$ which never increases the range $R$ beyond $\frac{R}{2^{n-1}}$.
Similarly, given two mother QCA $\alpha_0,\beta_0$ connected by a continuous path of $R$-QCA, there is a continuous path of 
$\frac{R}{2^n}$-QCA from $\alpha_n$ to $\beta_n$.

Thus the different terminal sequences with mother $\lbar{\alpha}$ and range $\frac{R}{2^n}$ are all equivalent: the sequence, in fact, provides no new information and the naive definitions based on individual QCA and bounded depth quantum circuits are essentially correct (given the existence of descendant sequences and their canonical nature).

The official definition of $Q(X) = \lbar{\text{QCA}}(X) / \lbar{\text{fdqc}}(X)$ goes through sequences, hence the bars, but we have seen the bars offer no new information when compared with the naive definitions. The naive definition may well be criticized on the ground that long composition eventually span the ``separate scales" and are not well defined. But they \emph{become} well defined when the existence and uniqueness of coherent sequences is added to the discussion.

In this discussion we have used 1-parameter families of coherent refinements to achieve well-defined group compositions of arbitrary length. Higher $k$-parameter families of refinements show the group law is also continuous when applied to $(k-1)$ parameter families of group elements. We do not explore the implications of higher families here.

{\bf Remark on Pleating Depth:} The pleating trick is an example of the Eilenberg Swindle. As we have explained, each step of the allowed deformations (1-4) is allowed under the definition of coherence.
However, it worth emphasizing that one further application of
the Swindle shows that, so long as we impose no a priori bound on the density of ancillas, the {\it depth} of the resulting
circuit is also $O(1)$.
Suppose that $\alpha_1$ and $\alpha_n$ are connected by a long sequence of QCA $\{\alpha_1, \dots, \alpha_n\}$ with each $\alpha_{i+1}$ derived from $\alpha_i$ by a fdqc. While at first glance the path from $\alpha_1$ to $\alpha_n$ seems to require depth $O(n)$, using the swindle $\alpha_1$ and $\alpha_n$ in fact can be connected by a depth $O(1)$ circuit. An initial circuit of depth $O(1)$ can be created from the identity as
\[
(\alpha_1^{-1} \otimes \alpha_2) \otimes (\alpha_3^{-1} \otimes \alpha_4) \otimes \cdots \otimes (\alpha_{n-1}^{-1} \otimes \alpha_n)
\]

Regrouping and canceling using a second depth $O(1)$ circuit yields
\begin{align*}
	\alpha_1^{-1} \otimes (\alpha_2 \otimes \alpha_3^{-1}) & \otimes (\alpha_4 \otimes \alpha_5^{-1}) \otimes \cdots \otimes \alpha_n \\
	& \downarrow \\
	\alpha_1^{-1} & \otimes \alpha_n
\end{align*}
which can then be used to pass from $\alpha_1$ to $\alpha_1 \otimes (\alpha_1^{-1} \otimes \alpha_n)$ to $\alpha_n$, all with depth $O(1)$.

Thus the 4 operations (lines \ref{operations}) suffice to define a single-shot equivalent relation on QCA adequate for all our theorems, provided an unbounded ancilla density is permitted. If it is not, the alternative definition via deformation is to be performed.

We conclude this section with a comment on normality. When deforming a quantum circuit or even a QCA there is no difference between allowing deformation at the beginning, middle or end of the circuit (or QCA). For example, suppose we wish to insert a swap $s$ (itself deformable from identity) between $\alpha$ and $\alpha^{-1}$ we can do this in three equivalent ways:
\begin{equation}
	(\alpha s \alpha^{-1})\alpha \alpha^{-1} = \alpha s \alpha^{-1} = \alpha \alpha^{-1}(\alpha s \alpha^{-1})
\end{equation}
For this reason we have \emph{not} been explicit in the main text as to where deformation occurs.

\textbf{Product of Families Equals Family of Product:}
Theorem \ref{deform} also implies another result: if $\lbar{\alpha},\lbar{\beta},\lbar{\gamma}$ are coherent families, with $\alpha_0 \circ \beta_0=\gamma_0$, then, using ancillas, there is a continuous path from $\alpha_n \circ \beta_n$ to $\gamma_n$.
To prove this, note that there is a continuous path of $R_n$-QCA from $(\alpha_n \circ \beta_n)\otimes \Id$ to
$\alpha_n \otimes \beta_n$,
where we have added one ancilla degree of freedom for each physical degree of freedom.
Then, there is a deformation from $\alpha_n\otimes \beta_n$ to $\alpha_0\otimes \beta_0$, a deformation from $\alpha_0 \otimes \beta_0$ to $\gamma_0 \otimes \Id$, and finally a deformation from $\gamma_0 \otimes \Id$ to $\gamma_n \otimes \Id$.
Composing these formations and applying theorem \ref{deform} gives the desired result.

\section{Coherent Families Using Circuits}
\label{circ}
The previous section showed how to construct a coherent family given a mother QCA.  Theorem \ref{deform} then showed that given two path coherent familes, $\lbar{\alpha},\lbar{\beta}$ with $\alpha_0,\beta_0$, connected by a path of $R$-QCA, then $\alpha_n,\beta_n$ were connected by a path of $\frac{R}{2^n}$-QCA.  In this section we show a similar result for {\it circuit-coherent} families.

First let us show the claim that a circuit with range bounded by $R$ (or more generally the forward lightcone of every qudit having diameter bounded by $R$) can be written in depth $d+1$ using gates of range $O(R)$, assuming an appropriate handle decomposition on a $d$-dimensional control space.
To show this, decompose the qudits into disjoint sets $S_0,S_1,\ldots,S_d$ (in this case, each set $S_j$ will represent the qudits in the union of $j$-handles).
Let $V_j$ be the unitary implemented by the product of gates in the quantum circuit which are supported on $S_j$ and which are not in the forward lightcone of any $S_k$ for $k>j$; here we take the product of gates in the obvious order.
Then, if the circuit implements unitary $U$ we have $U=V_d V_{d-1} \ldots V_0$.
Assuming each $S_j$ is a union of sets of diameter $O(R)$, with distance between the sets larger than $10R$, then each $V_j$ can be regarded as a one-round quantum circuit using gates of diameter $O(R)$.
In this subsection, then, we will assume such a handle decomposition exists and treat all notions of range for a quantum circuit equivalently.

First, we emphasize that the construction of coherent families in the previous section gives a circuit-coherent family since each deformation can be implemented by circuits or stabilization.  Let us say that two QCAs $\alpha,\beta$ are ``equal up to a circuit of depth $\dep$ and range $R$" if $\alpha \circ \beta^{-1}$ is equal to such a circuit.
This definition raises
the following question: suppose
$\lbar{\alpha},\lbar{\beta}$ are two circuit coherent families, with $\alpha_0,\beta_0$ equal up to a circuit of depth $O(1)$ and range $O(R)$.  Is it the case that
$\beta_n,\alpha_n$ are equal to a circuit of depth $O(1)$ and range $O(\frac{R}{2^n})$?
A second question involves products: suppose $\lbar{\alpha},\lbar{\beta},\lbar{\gamma}$ are three circuit coherent families, with $\alpha_0 \circ \beta_0$ equal to $\gamma_0$ up to a circuit of depth $O(1)$ and range $O(R)$.  Is it the case that
$\alpha_n \circ \beta_n$ is equal to $\gamma_n$
 up to a circuit of depth $O(1)$ and range $O(\frac{R}{2^n})$?

In this section we answer both questions affirmatively.  Note first that in both cases it is straightforward to show a slightly weaker equivalence up to a circuit of depth $O(1)$ and range $O(R)$, as we sketch in the next two paragraphs.  We then prove a lemma \ref{circcompress} which shows that this equivalence implies equivalence up to a circuit of depth $O(1)$ and range $O(\frac{R}{2^n})$.  The key point which allows us to apply lemma \ref{circcompress} is that even though the circuit may
use gates of range $O(R)$, the circuit acts as a QCA of much smaller range, $O(\frac{R}{2^n})$.  For example, if
$\alpha_n,\beta_n$ are equal up to a circuit and $\alpha_n,\beta_n$ both have range $\frac{R}{2^n}$, then the circuit acts as a QCA with range at most $2\frac{R}{2^n}$.

Consider first the first question above, equivalence of $\beta_n$ and $\alpha_n$.  Note first that $\alpha_0$ is $O(R)$-circuit equivalent to $\alpha_n$ by definition of a circuit equivalent family: $\alpha_i$ and $\alpha_{i+1}$ are $O(\frac{R}{2^i})$-circuit equivalent, and summing over $i$ this gives an equivalence by a circuit with lightcone bounded by $O(R)$.
Similarly, $\beta_0$ is $O(R)$-circuit equivalent to $\beta_n$ and $\alpha_0$ is $O(R)$-circuit equivalent to $\beta_0$
by assumption.
Combining these equivalences gives the equivalence between $\alpha_n$ and $\beta_n$.  Of course, the equivalence between $\alpha_0$ and $\alpha_n$ may involve stabilization or destabilization, and similarly for the equivalence between $\beta_0,\beta_n$; if the resulting circuit from $\alpha_n$ to $\beta_n$ uses ancillas, these may be removed using the techniques of section \ref{arqc} given Assumption \ref{a1}.

For the second question, $\alpha_n$ is $O(R)$-circuit equivalent to $\alpha_0$; let $\mu$ be the QCA corresponding to this circuit.
Similarly, $\beta_n$ is $O(R)$-circuit equivalent to $\beta_0$; call the corresponding QCA $\nu$.
Then, $\alpha_n \circ \beta_n=(\mu \circ \alpha_0) \circ (\nu \circ \beta_0)$.  However, $\alpha_0 \circ \nu=\nu' \circ \alpha_0$ where $\nu'=\alpha_0 \circ \nu \circ \alpha_0^{-1}$ is a quantum circuit of range $O(R)$.
So, $\alpha_n \circ \beta_n$ is equivalent to $\gamma_0$ up to a circuit of range $O(R)$; combining with the equivalence between $\gamma_0$ and $\gamma_n$ gives the desired result.

Finally we prove the lemma.
\begin{lemma}
\label{circcompress}
Let $C$ be some circuit, with the range of the lightcone of the circuit bounded by $R$, for some sufficiently small $R$.
 Assume  that
the circuit acts as a QCA with range $r \ll R$.
Then, there exists a circuit $C'$ that implements the same unitary as $C$, 
with the range of the lightcone of $C'$ bounded by $r$.
\begin{proof}
Let $\alpha_C$ denote the QCA corresponding to circuit $C$.
We will apply the pleating construction above to construct a pleated QCA $\alpha_{C_P}$.
Previously, we applied the pleating construction to reduce the range of a QCA.  However, the pleating construction also reduces the range of the gates in the circuit: each occurrence of $\alpha_C$ or $\alpha_C^{-1}$ in the pleating construction away from the folds can be decomposed into gates of $C$ or $C^{-1}$, and the deformation of the pleat reduces the range of these gates.
With this decomposition there is a ``pleated circuit" $C_P$ which corresponds to QCA $\alpha_{C_P}$.

We construct $C_P$ as follows.
For each occurrence of $\alpha$ or $\alpha^{-1}$ in the pleating construction (i.e., for each sheet), we decompose the corresponding circuit $C$ or $C^{-1}$ as
a product $C_{near} \circ C_{far}$ where $C_{near}$ includes gates near the folds, i.e., within distance $O(R)$ of the folds, and $C_{far}$ includes the remaining gates far from the folds which are not in the forward lightcone of $U_{near}$.
Then we write $C_P=C_{P,near} \circ C_{P,far}$ where $C_{P,far}$ is obtained in the obvious way from the gates in each $C_{far}$, i.e., for each sheet in the pleat, we insert the gates of $C_{far}$ for that sheet.
To construct $C_{P,near}$, recall the construction near the folds for an arbitrary QCA $\alpha$.  The product $\alpha \otimes \alpha^{-1}$ near the fold is a quantum circuit, as  we write $\alpha \otimes \alpha^{-1}=\Bigl((\alpha \otimes \Id) \circ \Swap \circ (\alpha^{-1} \otimes \Id)\Bigr) \circ \Swap$, where $\Swap$ swaps degrees of freedom in the two tensor factors.
The operator $\Bigl((\alpha \otimes \Id) \circ \Swap \circ (\alpha^{-1} \otimes \Id)\Bigr)$ is a circuit for any $\alpha$, as is 
$\Swap$, so this gives the circuit near the folds.

We will choose the pleats so that the range of the lightcone in $C_P$ is bounded by $O(\frac{R}{2^n})$ for $n$ large enough that $\frac{R}{2^n}<r$.  We can regard $\alpha_C$ as the mother of a coherent family with $\alpha_{C_P}$ being its $n$-th member.  Hence, since $\alpha_C$ is a QCA of range $r$, by definition of a coherent family there is a circuit $D$ of lightcone range $O(r)$
from $\alpha_{C_P}$ to $\alpha_{C}$, i.e., $D$ acts as $\alpha_C \circ \alpha_{C_P}^{-1}$.
Finally, let
\begin{align}
C'=D \circ C_P.
\end{align}
\end{proof}
\end{lemma}

\section{Families of Translation Invariant QCA}
\label{cft}

In this section the control space is a torus unless noted otherwise.
The set of sites will form a hypercubic lattice
with $L_1 \times L_2 \times \cdots \times L_d$ sites.
The distance between neighboring sites will be $L_j^{-1}$ along $j$-th direction.
A translation invariant (TI) QCA
is completely specified by the images of single qudit operators at one site and the numbers $L_j$;
the TI QCA is invariant under any translation by one site along any of $d$ directions.
Thus, given a TI QCA $\alpha_0$ of range $R \le 1/100$ on, say, $100^d$ sites,
the ``obvious'' family of a TI QCA is determined by a sequence of numbers of sites.

In Section~\ref{sec:def} we pointed out that the obvious family 
of a TI QCA on $d$-torus may not be coherent.
Given a $3$-dimensional TI nontrivial QCA $\gamma$ of order~4,
we considered an example that is a $4$-dimensional TI QCA $\bar \beta$ constructed by stacking $\gamma$.
We tuned the boundary conditions ($L_j$)
such that the number of $\gamma$ in $\beta_i$ does not conform with the order of $\gamma$,
rendering $\bar \beta$ incoherent.
However, in this example a \emph{sub}sequence of $\bar \beta$ \emph{is} coherent
as one can insert multiples of four layers of $\gamma$ by a quantum circuit.
So, if we are to make a general claim on the coherence of a TI QCA,
the best we should aim for is to show that some subsequence of the obvious family is coherent.
Here we give a result in that direction.

Note that there is an open possibility that there exists a TI QCA $\eta$ of infinite order,
i.e., $\eta^n$ is a quantum circuit if and only if $n = 0$.
If such $\eta$ exists, then the stack of $\eta$, which is TI, will not contain any coherent subsequence.
In addition, the hypothetical $\eta$ of infinite order makes two notions of equivalence of QCA distinct:
One notion of equivalence is by quantum circuits that we are considering in this paper,
and the other is by ``blending''~\cite{fh} --- two QCAs $\alpha$ and $\beta$ 
of range $R$ blend between disjoint regions $A$ and $B$
if there is an interpolating $R$-QCA that agrees with $\alpha$ on $A$ and $\beta$ on $B$.
The circuit equivalence implies the blending equivalence.
A stack of some QCA always blends with the identity QCA by terminating the stack,
but the full stack of $\eta$ that has infinite order cannot be a circuit.
The potential distinction between blending equivalence and circuit equivalence
has a footprint in the coherence of TI QCA as we will see now.

\begin{defin}
Let $I = [-1,1]$ and $S^1 = \{ (x,y) \in \mathbb{R}^2 : x^2 + y^2 = 1\}$.
They admit an orientation-reversing diffeomorphism $\ocinv$ by interchanging $x$ with $-x$.
If $\alpha$ is a QCA on $X = X' \times I$ or $X = X' \times S^1$ 
where the geometry of sites is invariant under $\ocinv$ acting on $I$ or $S^1$,
we define $\alpha^*$ to be $\ocinv \circ \alpha \circ \ocinv$, called one-coordinate inversion of $\alpha$.
\end{defin}
Note that for any TI QCA $\alpha$, the product $\alpha \otimes \alpha^*$ blends into the identity:
Just project $S^1$ in $T^d$ down to $I$ by $(x,y) \mapsto x$.
If the blending equivalence implied circuit equivalence, 
then we would always have $\alpha^* \cong \alpha^{-1}$.

In the pleating construction for \emph{the} coherent family of QCA starting with a mother TI QCA,
specialized to the situation where the pleats are all parallel to a coordinate axis,
one half of the added degrees of freedom is acted on by $(\alpha^{-1})^* =: \alpha^{-*}$
and the other half by $\alpha$,
except for ``cusps'' of the pleats.
If the obvious family of some TI QCA is coherent, then, intuitively, 
$\alpha^{-*}$ has to match $\alpha$ up to a finite depth quantum circuit.
Since $\alpha^{-*}$ blends with $\alpha$,
if the blending equivalence implied circuit equivalence,
then we may expect that the obvious family of TI QCA would always be coherent.
Not knowing this, we instead prove the following theorem, which still bears an interesting corollary.

\begin{figure}
\centering
\includegraphics[width=\textwidth, trim = {130 230 70 0}, clip]{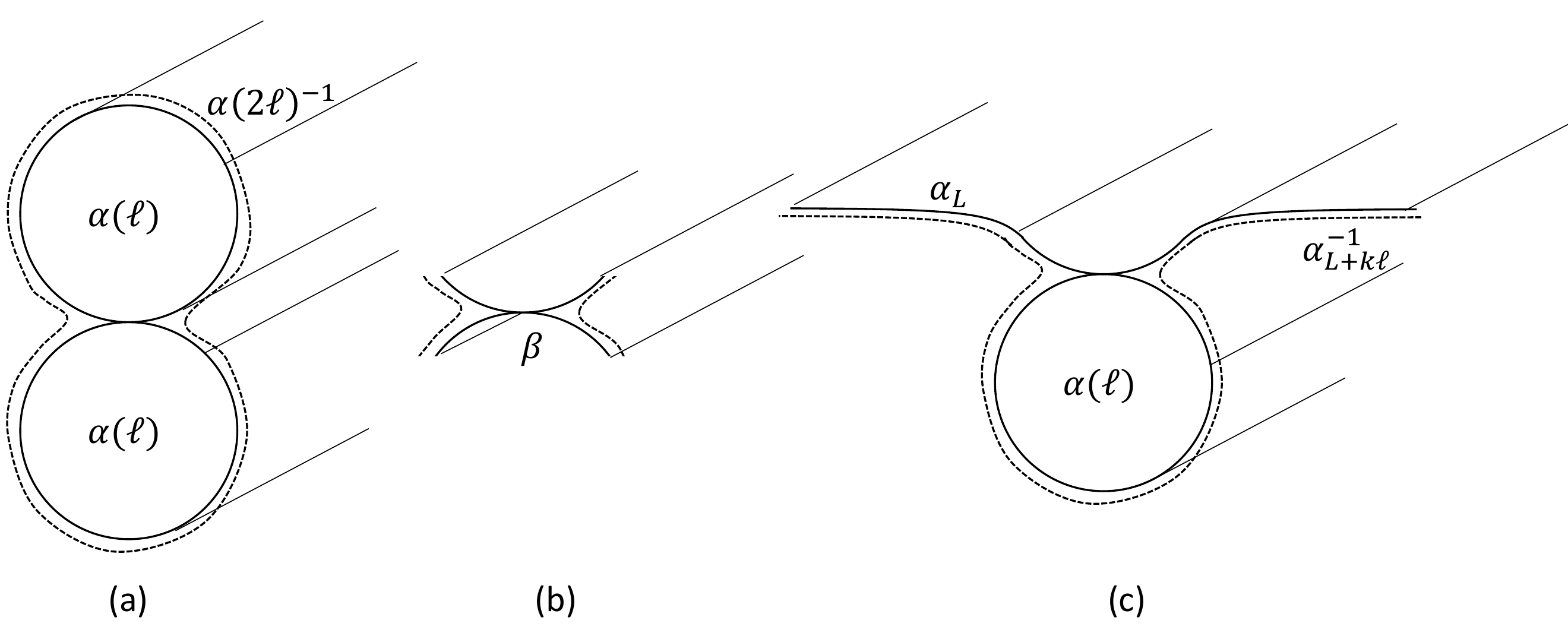}
\caption{Construction of translation invariant pleats. 
Assuming that the dimensionally reduced QCA has finite order in the group of QCA modulo fdqc,
the $k$-th power of~(a) is trivial for some $k < \infty$.
This implies that $k$-th power of~(b) is trivial.
In the specific pleating of~(c), $\beta$ appears exactly $k$ times,
implying the coherence of the obvious family of translation invariant QCA.
}
\label{fig:TIcoherence}
\end{figure}

\begin{theorem}\label{thm:ti-qca-coherence}
Let $\bar \alpha = \{ \alpha_L \}$ on $L^d$ sites ($L = L_0, L_0 + 1, L_0 + 2,\ldots$)
be the obvious family of a TI QCA on an $d$-torus.
Here, $L_0$ is so large that the range of $\alpha$ is at most $1/10$ of the linear size of the torus.
Using the same local action as $\bar \alpha$,
for $j = 1,2,\ldots, d$
we define $\overline{\alpha(\ell,j)} = \{ \alpha_L(\ell,j) \}$ on $L^{j-1} \times \ell \times L^{d-j}$ sites 
to be another obvious family of TI QCA on an $(d-1)$-torus,
obtained by dimensionally reducing in the $j$-th coordinate. 
If for any $j$ and for any sufficiently large $\ell$, 
$\overline{\alpha(\ell,j)}$ has finite order 
in the group of QCA modulo finite depth quantum circuits,
then a subsequence of $\bar \alpha$ is coherent.
\end{theorem}
\begin{proof}
We claim that there is a deformation $\alpha_L$ into $\alpha_{L'}$ by stabilization and finite depth quantum circuits
where $L' - L$ is any multiple of some fixed $\ell_0$.
We do this by increasing the number of sites along one coordinate after another.
This is enough for the theorem.

Choose $\ell$ such that it is at least $10^3$ times larger than the range $R$ of $\alpha_L$.
Let $\overline{\alpha(\ell,j)}$ have order $k' < \infty$ 
and let $\overline{\alpha(2\ell,j)}$ have order $k'' < \infty$.
Then, both $\overline{\alpha(\ell,j)}$ and $\overline{\alpha(2\ell,j)}$ 
become trivial at the $k$-th power where $k \le k' k''$.
We will prove the claim with $\ell_0 = k \ell$.

Consider a control space $X = \text{``8''} \times T^{d-1}$.
The figure ``8'' is the $j$-th direction along which there are exactly $2\ell$ sites,
and $T^{d-1}$ is a torus.
For clarity of presentation, 
let us distinguish the one circle that goes around the outer rim of the figure ``8''
from the two disjoint circles in the inner rims.
Let $\overline{\alpha(2\ell,j)^{-1}}$ be defined along the outer rim,
and two $\overline{\alpha(\ell,j)}$ along the inner rims. See Figure~\ref{fig:TIcoherence}(a).

Then, except for the $O(R)$-neighborhood of the middle crossing point of the figure ``8'',
the composed action of the three QCA's is the identity.
Hence, we are left with a TI QCA $\bar \beta$ on the $(d-1)$-torus at the middle point of the figure ``8''.
See Figure~\ref{fig:TIcoherence}(b).
Since $\overline{\alpha(2\ell,j)}^k \cong \Id$ and $\overline{\alpha(\ell,j)}^k \cong \Id$,
we know $\bar \beta^k \cong \Id$.

Now, bring one $\alpha_L$ and $k$ copies of $\alpha_L(\ell,j)$ that are laid adjacent to $\alpha_L$.
See Figure~\ref{fig:TIcoherence}(c).
Wrap the arrangement with $\alpha^{-1}$ on $L^{j-1} \times (L + k\ell) \times L^{d-j}$ sites.
We see that there is $\beta_L$ at each of the $k$ singular points where $(d-1)$-torus is,
which are jointly trivial.
\end{proof}

\begin{corollary}
For any $d = 0,1,2,\ldots$,
every obvious family of $d$-dimensional TI Clifford QCA with prime dimensional qudits contains an infinite coherent subsequence.
(A Clifford QCA is one that maps every Pauli matrix to a tensor product of Pauli matrices.)
Every obvious family of $3$-dimensional TI QCA contains an infinite coherent subsequence.
\end{corollary}
\begin{proof}
In any dimension, 
every TI Clifford QCA with prime dimensional qudits has order 1, 2, or 4~\cite{haah2019clifford}.
Every $2$-dimensional QCA has order~1~\cite{fh}.
\end{proof}

Remark:
The above argument applies verbatim to any ``invertible states''~\cite{Kapustin_cobordism,Freed}
that appear in mathematical many-body physics.
One can consider an equivalence relation of the states 
by finite depth quantum circuits where each gate is $G$-symmetric 
for some (possibly trivial) group $G$.
Then, any translation invariant invertible state is coherent
if it becomes finite order upon dimensional reductions 
(compactification along any one direction).
In more physical terms, such a state is automatically 
an entanglement renormalization group fixed point.
Similarly, the construction of section \ref{sec:cf} can be applied to any such invertible state.

\appendix
\section{QCA Are an Abelian Group Modulo Circuits}

Here we recall the proof that QCA are an abelian group modulo circuits\cite{fh,hfh}
. The argument is summarized in Fig.~\ref{figabelian} below.

	\begin{figure}[ht]
		\centering
		\begin{tikzpicture}[scale = 0.8]
		\draw (3, 3) -- (3, -3);
		\node at (3.9,0) {$\equiv$};
		\draw  (0.6,1.7) rectangle (2,1);
		\draw  (0.6,-1) rectangle (2,-1.7);
		\draw (1.3, 3) -- (1.3, 1.7);
		\draw (1.3, 1) -- (1.3, -1);
		\draw (1.3, -1.7) -- (1.3, -3);

		\draw  (4.6,1.7) rectangle (6,1);
		\draw  (4.6,-1) rectangle (6,-1.7);
		\draw (4.9,3) to [out = -90, in = 135] (5.3,1.7);
		\draw (5.7,1) to [out = -45, in = 90] (7.2,-0.6) to [out = -90, in = 90] (5.3,-3);
		\draw (7,3) to [out = -90, in = 45] (6.3, .8);
		\draw (5.9, 0.5) to [out = 225, in = 90] (5,-0.5) to [out = -90, in = 125] (5.2,-1);
		\draw (5.7, -1.7) to [out = -45, in = 155] (5.9, -1.9);
		\draw (6.1, -2.1) to [out = -45, in = 90] (6.8, -3);

		\draw (9, 3) -- (9, 1.7);
		\node at (8.1,0) {$=$};
		\draw  (8.3,1.7) rectangle (9.7,1);
		\draw (9, 1) -- (9, -3);
		\draw (10.7, 3) -- (10.7, -1);
		\draw  (10,-1) rectangle (11.4,-1.7);
		\draw (9, 1) -- (9, -3);
		\draw (10.7, -1.7) -- (10.7, -3);
		\node at (11.7,0.1) {$=$};

		\draw (13.2,3) -- (13.2,0.4);
		\draw (13.2,-0.3) -- (13.2,-3);
		\draw (15,3) -- (15,0.4);
		\draw (15,-0.3) -- (15,-3);
		\draw  (12.5,0.4) rectangle (13.9,-0.3);
		\draw  (14.3,0.4) rectangle (15.7,-0.3);

		\node at (1.3,1.3) {$\alpha_1$};
		\node at (1.3,-1.4) {$\alpha_2$};
		\node at (5.3,1.3) {$\alpha_1$};
		\node at (5.3,-1.4) {$\alpha_2$};
		\node at (9,1.3) {$\alpha_1$};
		\node at (10.7,-1.4) {$\alpha_2$};
		\node at (13.2,0) {$\alpha_1$};
		\node at (15,0) {$\alpha_2$};
		\end{tikzpicture}
		\caption{}
\label{figabelian}
	\end{figure}

	The left vertical line represents $\mathcal{H} = \otimes_i \mathcal{H}$ over $\{x_i\} \subset X$ (that is, $X$ is imagined normal to the figure) and the stacking is composition of QCA, $\alpha_2 \circ \alpha_1$. The line to its right represents doubling the degrees of freedom by introducing for every $\mathcal{H}_i$ at $x_i$ an identical ancilla at the same location. The $\equiv$ sign represents equality in the quotient group $Q$. The crossings (under/over has no significance here) are our notation for the swap operators which at $x_i$ \textit{swaps} the local Hilbert space $\mathcal{H}_i$ with its ancillary partner. Notice that a swap is a fdqc of depth one, so may be added at will without changing the element in $Q$. At the far right we have the manifestly abelian group structure.

\bibliography{group-ref}
\end{document}